\PassOptionsToPackage{dvipsnames}{xcolor}
\documentclass[11pt]{article}
\usepackage{graphicx} 
\usepackage[utf8]{inputenc}
\usepackage{newpxtext}
\usepackage{fullpage}
\usepackage{booktabs} 

\usepackage{newpxtext}
\usepackage{amsmath,amssymb,amsthm}
\usepackage{newpxmath}
\usepackage{nicefrac}

\usepackage{booktabs} 
\usepackage[ruled,noline]{algorithm2e} 

\SetAlFnt{\small}
\SetAlCapFnt{\small}
\SetAlCapNameFnt{\small}
\SetAlCapHSkip{0pt}
\IncMargin{-\parindent}

 \usepackage{amsmath,amsthm,graphicx,url}
\usepackage{amssymb}
\usepackage{threeparttable}
\usepackage{tikz}
\usepackage{subcaption}
 \usepackage{mathtools}
\usepackage[normalem]{ulem}
\usepackage{soul}
\usepackage{subcaption}
\usepackage{booktabs}
\usepackage{tablefootnote}
\usepackage[numbers]{natbib}
\usepackage{enumitem}
\usepackage{comment}
\DeclareMathOperator*{\argmax}{arg\,max}

 \usepackage[ruled, noline]{algorithm2e} 

 \SetKwBlock{Initialize}{Initialize:}{} \SetKwFor{From}{from}{do}{}
 \SetKwFor{ForAllInParallel}{for all}{in parallel do}{}

\SetCommentSty{mycommfont}

\usepackage{hyperref}
\usepackage[nameinlink, capitalise,noabbrev]{cleveref}

\usepackage[dvipsnames]{xcolor}
\definecolor{niceRed}{RGB}{190,38,38}
\definecolor{blueGrotto}{HTML}{059DC0}
\definecolor{royalBlue}{HTML}{057DCD}
\definecolor{navyBlue}{HTML}{0B579C}
\definecolor{limeGreen}{HTML}{81B622}
\definecolor{nicePurple}{HTML}{9c27b0}
\definecolor{lightRoyalBlue}{HTML}{def2ff} 
 \definecolor{gold}{HTML}{ffa300}
 \definecolor{commentblue}{HTML}{2A6EC7}
 \definecolor{DrexelBlue}{HTML}{003058}

\definecolor{linkc}{rgb}{0.7, 0.2, 0.3}
\definecolor{citec}{rgb}{0.2, 0.3, 0.7}
\definecolor{urlc}{rgb}{0.2, 0.4, 0.3}
\hypersetup{
  colorlinks = true,  
  urlcolor = {commentblue},
  linkcolor = {DrexelBlue}, 
  citecolor = {YellowOrange!85!black}
}

\usepackage{thmtools} 
\usepackage{thm-restate}

\crefname{theorem}{Theorem}{Theorems}
\crefname{lemma}{Lemma}{Lemmas}
\crefname{conjecture}{Conjecture}{Conjectures}
\crefname{corollary}{Corollary}{Corollaries}
\crefname{proposition}{Proposition}{Propositions}
\crefname{claim}{Claim}{Claims}
\crefname{observation}{Observation}{Observations}
\crefname{definition}{Definition}{Definitions}
\crefname{example}{Example}{Examples}
\crefname{remark}{Remark}{Remarks}
\crefname{question}{Question}{Questions}
\crefname{condition}{Condition}{Conditions}
\crefname{assumption}{Assumption}{Assumptions}

\SetAlFnt{\small}
\SetAlCapFnt{\small}
\SetAlCapNameFnt{\small}
\SetAlCapHSkip{0pt}
\IncMargin{-\parindent}

\makeatletter
\g@addto@macro\@floatboxreset\centering
\makeatother

\newcommand{\xtnote}[1]{}

\newtheorem{theorem}{Theorem}[section]
\newtheorem{lemma}[theorem]{Lemma}
\newtheorem*{theorem*}{Theorem}
\newtheorem{corollary}[theorem]{Corollary}

\newtheorem{remark}[theorem]{Remark}
\newtheorem{definition}[theorem]{Definition}
\newtheorem{observation}[theorem]{Observation}
\newtheorem{mechanism}{Mechanism}

\DeclareMathOperator{\E}{\mathbb{E}}
\newcommand{\R}{\mathbb{R}}

\newcommand{\bids}{\mathbf{b}}
\newcommand{\sbids}{\mathbf{s}}

\newcommand{\eps}{\varepsilon}

\newcommand{\alloc}{\mathbf{x}}
\newcommand{\gft}{\texttt{GFT}}

\newcommand{\pb}{\mathbf{p}^b}
\newcommand{\ps}{\mathbf{p}^s}
\newcommand{\wel}{\texttt{Wel}}
\newcommand{\poa}{\texttt{PoA}}
\newcommand{\welpoa}{\texttt{Wel-PoA}}
\newcommand{\gftpoa}{\texttt{GFT-PoA}}
\newcommand{\val}{\mathbf{v}}
\newcommand{\cost}{\mathbf{c}}
\newcommand{\tbu}{\texttt{TB}_u}
\newcommand{\tbv}{\texttt{TB}_v}
\newcommand{\tbo}{\texttt{TB}_o}
\newcommand{\rsu}{\texttt{RS}_u}
\newcommand{\rsv}{\texttt{RS}_v}
\newcommand{\rso}{\texttt{RS}_o}

\newcommand{\BuyerWelfare}{\textsc{BuyerWel}}
\newcommand{\BuyerRevenue}{\textsc{BuyerPayment}}
\newcommand{\SellerWelfare}{\textsc{SellerWel}}
\newcommand{\SellerRevenue}{\textsc{SellerRev}}
\newcommand{\TwoSidedWelfare}{\textsc{TwoSidedWel}}

\numberwithin{equation}{section}

\title{Two-sided Market Design Meets Autobidding}
\date{}

\author{
  \begin{tabular}{c@{\hspace{3em}}c@{\hspace{3em}}c}
    Yang Cai & Christopher Liaw & Aranyak Mehta \\
    Yale University & Google Research & Google Research \\
    \small\href{mailto:yang.cai@yale.edu}{yang.cai@yale.edu} & 
    \small\href{mailto:cvliaw@google.com}{cvliaw@google.com} & 
    \small\href{mailto:aranyak@google.com}{aranyak@google.com}
  \end{tabular} 
  \\[3.5ex] 
  \begin{tabular}{c@{\hspace{4em}}c}
    Xizhi Tan & Mingfei Zhao \\
    Stanford University & Google Research \\
    \small\href{mailto:xizhi@stanford.edu}{xizhi@stanford.edu} & 
    \small\href{mailto:mingfei@google.com}{mingfei@google.com}
  \end{tabular}
}

\begin{document}


\maketitle

\begin{abstract}
Autobidding has become a dominant paradigm in online advertising by enabling advertisers to set high-level goals while algorithms handle real-time bid optimization. A prominent example is the Return-on-Spend (RoS) \emph{value maximizer}, which maximizes total value subject to an aggregate value-per-spend constraint. While most work on mechanism design for autobidders focuses on one-sided markets, many real-world platforms involve strategic behavior on both sides.

We initiate the study of \emph{two-sided markets with autobidders}. We first establish a stark negative result: for the challenging objective of \emph{liquid gains from trade} (LGFT) in the prior-free setting, broad classes of utility-truthful, budget-balanced mechanisms have unbounded Price of Anarchy (PoA) once value-maximizing agents are present, even in double-auction environments.

We then give two positive results. In the prior-free repeated double-auction setting, McAfee's Trade Reduction mechanism has unbounded PoA for LGFT but achieves constant PoA for \emph{liquid welfare} under RoS bidding, showing that off-the-shelf mechanisms retain meaningful welfare guarantees. With distributional information regarding the private costs and values, we design a two-sided mechanism that is incentive compatible for each agent under its corresponding objective, individually rational, ex-ante weakly budget balanced, and achieves first-best liquid welfare, equivalently optimal LGFT, whenever at least one side of the market consists of RoS value maximizers. This result applies to matching markets with general downward-closed feasibility constraints. It contrasts sharply with the Myerson--Satterthwaite impossibility theorem~\citep{MS83}, which rules out first-best efficiency with incentive compatibility, individual rationality, and budget balance even in bilateral trade when both the buyer and the seller are quasi-linear utility maximizers.

\end{abstract}

\thispagestyle{empty}
\setcounter{page}{0}

\section{Introduction}
Autobidding systems are widely used in online advertising, with many advertisers now relying on them. Traditionally, an advertiser was required to manage detailed bids for each keyword group or ad impression. By contrast, with autobidding, the advertiser simply specifies an overall goal and broad constraints, and the autobidding agent leverages predictive models to adjust per-query bids in real time. This automated process simplifies an advertiser's workflow and frequently improves performance, often leading to significant increases in key metrics in practice.

Several forms of autobidding products exist. A common strategy is to maximize the total number of clicks or conversions subject to a specified return on spend (ROS) constraint. Here, the advertiser provides a target ``bang-per-buck'' ratio, and the autobidding algorithm procures ad slots so that the aggregated value and cost remain at least as good as the desired ratio. While autobidding makes campaign management easier, its growing popularity has changed how advertisers interact with ad platforms and introduced new challenges in mechanism design. 

A surge of research aims to understand the market with autobidders. Recent works have examined various facets such as bidding algorithms, equilibria induced by multiple interacting autobidders. One of the central—and still evolving—themes revolves around mechanism design: Do classical auction formats (originally designed for utility-maximizing advertisers) remain efficient with autobidding agents? (See \cite{aggarwal2019autobidding,deng2021towards,liaw2023efficiency,deng2022efficiency,deng2024non}.) Can one design new mechanisms that better align incentives and achieve better outcomes with autobidders? (See \cite{balseiro2021landscape, lv2023auction, Balseiro2024optimal, goel2014clinching,balseiro2022optimal,xing2023truthful}.) We refer the readers to \cite{aggarwal2024autobidding} for a comprehensive overview of these developments.

\paragraph{Towards Two-Sided Markets with Autobidding}
Most prior works study \emph{one-sided markets} (auctions), focusing exclusively on advertiser (buyer-side) incentives while modeling publishers (seller-side) as the mechanism designer. However, many real-world platforms--such as stock exchanges, online advertising exchanges (e.g., Google’s DoubleClick), online marketplaces (e.g., Amazon, eBay), and ride-sharing services (e.g., Uber, Lyft)--involve \emph{two-sided interactions}, where both buyers and sellers act as self-interested agents.
In such settings, mechanisms must account for strategic behavior on both sides of the market while ensuring \emph{budget balance}, i.e., they do not run a deficit.

A seminal result by \cite{MS83} reveals a fundamental tension between the budget-balance constraint and efficiency. They prove that even in the simplest bilateral trade setting—one buyer and one seller with independent private valuations—no mechanism can simultaneously achieve Bayesian incentive compatibility (BIC), individual rationality (IR), ex-ante weak budget balance (WBB), and first-best efficiency (i.e., guaranteeing trade whenever the buyer’s value exceeds the seller’s cost). This highlights how the budget balance constraint fundamentally limits the achievable efficiency in two-sided markets\footnote{Without the budget balance constraint, the VCG mechanism is BIC, IR, and achieves the first-best efficiency.} and has, therefore, motivated a long line of research on approximately maximizing efficiency under this constraint (see~\citep{MCAFEE,BCWZ17,BlumrosenD21} and the references therein).

Our work initiates the study of \emph{two-sided mechanism design with autobidders}, where buyers (and later sellers) are value-maximizing subject to ROS constraints. It is a priori unclear how the ROS constraints will interact with the budget balance constraint, nor is it clear if previously efficient mechanisms will continue to be so. Specifically, we ask:
\begin{align*}
    &\emph{Are mechanisms designed for utility maximizers efficient for value maximizers?}\\   &\qquad\qquad\qquad\qquad\qquad\emph{If not, can we design efficient mechanisms for value maximizers?} \tag{*}
\end{align*}
\vspace{-0.7cm}

\subsection{Our Results}
We initiate the study of two-sided markets with ROS autobidders (also referred to as \emph{value maximizers}), aiming to understand how they transform the landscape of these markets. We focus on a single-dimensional matching market with a general downward-closed feasibility constraint on what trades between buyers and sellers are allowed. This captures the well-studied “Bilateral Trade” and “Double Auctions” as special cases. 

We adopt two measures of efficiency: \emph{Liquid Welfare}, defined as the total target-normalized value to all participants under the final allocation, and \emph{Liquid Gains from Trade} (LGFT), defined as the increase in liquid welfare. 
In the autobidding literature for auctions, liquid welfare serves as the efficiency metric, rather than traditional social welfare, because the latter is inapproximable when buyers face ROS constraints (\cite[\S 2.2]{aggarwal2024autobidding}). While full maximization of liquid welfare implies maximum LGFT, these objectives diverge in approximation settings. Specifically, LGFT is strictly more challenging to approximate.\footnote{For instance, if no trade occurs, the resulting liquid welfare might still represent a reasonable fraction of optimal liquid welfare (due to initial endowments), whereas the LGFT would be zero. Consequently, any mechanism guaranteeing an $\alpha$-approximation to LGFT necessarily guarantees an $\alpha$-approximation to liquid welfare, but not vice versa. Thus, LGFT represents a more stringent objective, offering a stronger notion of efficiency.}

We first consider a standard multi-round \emph{prior-free} setting from the autobidding literature. More specifically, we focus on \emph{double auctions}, a well-studied class of two-sided markets where any buyer can trade with any seller. We then study the celebrated Trade Reduction (TR) mechanism, introduced by \cite{MCAFEE}, which is truthful for utility maximizers and budget-balanced. We evaluate its performance using the \emph{price of anarchy} (PoA), which measures the worst-case ratio between the optimal efficiency and the efficiency achieved by the mechanism at equilibrium.

It is well-known that TR is a $(1+\frac{1}{r-1})$-approximation to the optimal GFT~\citep{MCAFEE} when both sides of the market are utility maximizers, where $r$ is the size of the smallest optimal matching among all rounds. This bound provides a constant factor approximation to the optimal GFT as long as $r>1$, a mild condition in practice. In sharp contrast, we demonstrate a broad limitation: regardless of the value of $r$, no dominant strategy incentive compatible (DSIC) (w.r.t. utility maximizers) and WBB mechanism---including TR---can achieve a bounded \emph{Price of Anarchy} (PoA) for LGFT with value maximizers in the prior-free setting. In the remainder of this paper, we provide two ways to alleviate this negative finding.

\begin{table}
  \centering
    \caption{The price of anarchy (PoA) of Trade Reduction (TR) mechanism with value‑maximizing buyers/sellers. ``Utility'' stands for quasi‑linear utility‑maximization, and ``Value'' for value‑maximization. $r$ is the size of the smallest optimal matching among all rounds.}
  \begin{threeparttable}
    \begin{tabular}{|c|c|c|}
      \hline
      Market Composition    & PoA of LGFT                                & PoA of Liquid Welfare                                         \\ 
      \hline
      Utility / Utility     & $1 + \frac{1}{r-1}$~(\cite{MCAFEE})       & $1 + \frac{1}{r-1}$~(\cite{MCAFEE})                    \\ 
      \hline
      Value / Utility       & $\infty$\tnote{a}  \; (Thm.~\ref{thm:poa_gft_lowerbound_truthful_randomized})
                            & {$2 + \tfrac{1}{r-1}$} (Thm.~\ref{thm:one_side_value_poa})
      \\ 
      \hline
      Value / Value         & $\infty$\tnote{b}\; (Thm.~\ref{thm:poa_gft_lowerbound_truthful})
                            & $\bigl[3 - \tfrac{2}{r+1},\,3 + \tfrac{1}{r-1}\bigr]$ (Thm.~\ref{thm:two_side_value_poa})
      \\ 
      \hline
    \end{tabular}
    \begin{tablenotes}
      \footnotesize
      \item[a] This lower bound  holds for \emph{all} truthful mechanisms, not just TR.
      \item[b] This lower bound  holds for \emph{all deterministic} truthful mechanisms, not just TR.
    \end{tablenotes}
  \end{threeparttable}
  \label{tab:poa-tr}
  \vspace{-0.5cm}
\end{table}

Given that LGFT is a highly demanding measure, the first way we circumvent the above negative result is by considering the less stringent efficiency measure of Liquid Welfare. We show that the TR mechanism guarantees good Liquid Welfare at equilibrium. Specifically, if only one side of the market (either sellers or buyers) comprises value maximizers, we prove that the PoA for Liquid Welfare is $2+O(1/r)$. For settings with value maximizers on both sides, TR achieves a PoA of $3+O(1/r)$. A summary of these  results appears in \cref{tab:poa-tr}.

We also compare against the Liquid Welfare TR would achieve with truthful reporting (i.e., reporting true target-normalized values). Under this benchmark, we prove improved ratios of $2$ and $3$ for the single-sided and two-sided settings, respectively, which are independent of $r$. In each case, we provide matching lower bounds. 

To overcome the prior-free impossibility, our second approach explores scenarios where the designer has distributional information regarding the agents' values and costs. In contrast to \cite{MS83}, who showed that first-best efficiency is unattainable by any BIC, IR, and WBB mechanism for utility maximizers, we establish that first-best is achievable under these same constraints, when agents on one or both sides of the market are value maximizers. Crucially, this result extends beyond double auctions to general matching markets with arbitrary downward-closed feasibility.

\begin{theorem*} In a matching market where one or both sides comprise value-maximizing agents, there is a BIC, IR, and WBB mechanism that achieves the first-best liquid welfare.
\end{theorem*}

\noindent 
Interestingly, we observe that when both sides are value maximizers, the broker's expected profit equals
the expected first-best LGFT.
\subsection{Related Work}
\paragraph{Autobidding and Auction Design}
The problem of value-maximizing autobidding was initiated by \cite{aggarwal2019autobidding} which proposed the now widely adopted mathematical framework for autobidding.
Since then, there have been several different threads of work on the autobidding problem.
The subject closest to our work is in designing mechanisms when the agents are value-maximizers.
In their initial paper, \cite{aggarwal2019autobidding} showed that the VCG mechanism, originally designed for utility maximizers, has a price of anarchy of $2$ with value maximizers. Since then, several lines of work that focus on the price of anarchy of canonical auction formats as well as their variants. \cite{deng2021towards} generalize the bound of $2$ to multi-slot settings. Beyond truthful auctions, \cite{liaw2023efficiency} proves that the PoA of the first price auction is $2$, and
\cite{deng2022efficiency} generalizes the result to randomized strategies. Our PoA result for the TR mechanism is in the same spirit, but for double auctions. There have been many follow-up works which provide improved mechanisms for value-maximizing autobidders under various settings \cite{mehta2022auction,liaw2023efficiency,liaw2024efficiency,balseiro2021robust,deng2021towards,deng2023efficiency,deng2022efficiency,deng2022individual}. 
These works generally assume that autobidders only provide bids as input although some works, such as \cite{balseiro2021robust,deng2021towards}, show how to obtain improved welfare guarantees when the auction has noisy predictions of the autobidders' values.
A closely related work to ours in the one-sided autobidding setting is optimal mechanism design where the mechanism receives targets and values as inputs from the agents \cite{balseiro2021landscape}. 
Among other settings, \cite{balseiro2021landscape} show that for value-maximizing buyers with public targets
and private values, a payment-scaled second-price auction achieves first-best
revenue by making RoS constraints bind; the analogous first-best benchmark
is unattainable for utility maximizers. Part of our work can be viewed as generalizing these results to two-sided markets.  Later work considered optimal mechanism design under difference information settings \citep{lv2023auction,balseiro2023optimal,balseiro2022optimal,xing2023truthful}

Additional threads in the autobidding literature look at multi-channel autobidding where a single autobidder may be bidding on different channels that each have different auction formats (e.g.~\cite{paes2020competitive, aggarwal2023multi, aggarwal2024multi, susan2023multi, deng2023multi}) and designing bidding algorithms (e.g.~\cite{aggarwal2024no,feng2023online,balseiro2023joint,lucier2023autobidders,aggarwal2025multiplatform}). We refer the reader to \cite{aggarwal2024autobidding} for a more comprehensive survey on the autobidding literature.

\paragraph{Two-Sided Markets}
Motivated by the impossibility result of \cite{MS83}, one line of research in the computer science literature has been focused on designing mechanisms that approximate the optimal welfare or gains-from-trade in the bilateral trade setting \cite{BlumrosenD21,colini2017fixed,DengMSW22,KangPV22,CaiW23,BCWZ17,HartlineWang25,DengMSWW25,Jo26,LiuQRW26} and in more general settings such as double auctions and multi-dimensional two-sided markets \cite{Colini-Baldeschi16,colini2020approximately,DuttingRT14,BabaioffCGZ18,CaiGMZ21,MCAFEE,BRTW26,RTZ26,BeiLWZ26,BWZ26}. More recently, there has been some work on understanding how competition can help improve the efficiency of simple mechanisms, such as Trade Reduction \cite{BabaioffGG20,cai2024power} and on online bilateral trade \cite{azar2022alpha,bernasconi2024no,bolic2023online,cesa2024bilateral,cesa2021regret,cesa2023repeated,bachoc2024fair}.


\section{Preliminaries}\label{sec:prelim}
We consider a repeated two-sided market with \(T\) rounds, \(n\) buyers, and \(m\) sellers. 
For each round \(t\in[T]\), let \(B_t\subseteq[n]\) and \(S_t\subseteq[m]\) denote the buyers and sellers participating in that round. 
Each buyer \(i\in B_t\) is unit-demand and has a private value \(v_i^t\) for receiving the item or service. 
Each seller \(j\in S_t\) is unit-supply and has a private cost \(c_j^t\) for providing it.

For any round \(t\), a mechanism takes the buyers' bid vector \(\bids^t\) and sellers' bid vector \(\sbids^t\) as input, and outputs a potentially randomized allocation rule \(\alloc^t\) and payments. 
For every buyer \(i\) and seller \(j\), \(x_{ij}^t(\bids^t,\sbids^t)\) denotes the probability that buyer \(i\) trades with seller \(j\). 
Let \(p_i^t(\bids^t,\sbids^t)\) denote the payment collected from buyer \(i\), and let \(q_j^t(\bids^t,\sbids^t)\) denote the payment made to seller \(j\). 
For convenience, write
\[
    x_i^t(\bids^t,\sbids^t)=\sum_{j\in S_t}x_{ij}^t(\bids^t,\sbids^t),
    \qquad
    y_j^t(\bids^t,\sbids^t)=\sum_{i\in B_t}x_{ij}^t(\bids^t,\sbids^t).
\]
Thus \(x_i^t\) is the probability that buyer \(i\) trades in round \(t\), and \(y_j^t\) is the probability that seller \(j\) trades in round \(t\).

A buyer or seller is a \emph{quasi-linear utility-maximizer}, or simply a \emph{utility-maximizer}, if she maximizes total quasi-linear utility. 
For a buyer, quasi-linear utility is value minus payment; for a seller, it is payment minus cost. 
Since utility is additively separable across rounds and there are no cross-round constraints for utility-maximizers, maximizing total quasi-linear utility is equivalent to maximizing utility separately in each round. 
A mechanism is \emph{incentive compatible for utility-maximizers} (uIC) if truthful bidding maximizes every utility-maximizing agent's quasi-linear utility. 
A mechanism is \emph{individually rational} (IR) if truthful bidding gives every participating utility-maximizing agent nonnegative utility.

A buyer \(i\) is a \emph{value-maximizer} if she maximizes total received value subject to a return-on-spend (RoS) constraint across all rounds. 
Throughout the paper, RoS targets are public and values are private. 
Formally, buyer \(i\)'s problem is
\begin{align}
    \max_{\bids_i}\quad 
    &\sum_{t:i\in B_t} v_i^t x_i^t(b_i^t,\bids_{-i}^t,\sbids^t) \nonumber\\
    \text{s.t.}\quad 
    &\sum_{t:i\in B_t}
    \left(
        v_i^t x_i^t(b_i^t,\bids_{-i}^t,\sbids^t)
        - \tau_i p_i^t(b_i^t,\bids_{-i}^t,\sbids^t)
    \right)
    \ge 0 . \tag{RoS}
\end{align}
Here \(\tau_i\ge 1\) is buyer \(i\)'s public RoS target.

We also consider \emph{value-maximizing sellers}. 
A seller \(j\) is a value-maximizer if he maximizes the total cost of traded items subject to a revenue-over-cost (RoC) constraint across all rounds. 
This objective models sellers whose short-run goal is to maximize traded volume, utilization, or market share subject to a break-even or revenue adequacy constraint; see \cref{app:value_seller_motivation} for further discussion. 
We parameterize the RoC constraint by a public \(\tau_j\in(0,1]\), equivalently requiring revenue to be at least \(1/\tau_j\) times cost. 
Formally, seller \(j\)'s problem is
\begin{align}
    \max_{\sbids_j}\quad 
    &\sum_{t:j\in S_t} c_j^t y_j^t(\bids^t,s_j^t,\sbids_{-j}^t) \nonumber\\
    \text{s.t.}\quad 
    &\sum_{t:j\in S_t}
    \left(
        \tau_j q_j^t(\bids^t,s_j^t,\sbids_{-j}^t)
        - c_j^t y_j^t(\bids^t,s_j^t,\sbids_{-j}^t)
    \right)
    \ge 0 . \tag{RoC}
\end{align}
A mechanism is \emph{incentive compatible for value-maximizers} (vIC) if truthful bidding maximizes the corresponding value-maximizer objective among all bid vectors satisfying the agent's RoS or RoC constraint, holding the other agents' bids fixed. 
We use this terminology only when explicitly discussing truthful mechanisms for value-maximizers. 
In the prior-free PoA sections, value-maximizing agents are not assumed to bid truthfully; instead, outcomes are evaluated at pure Nash equilibria of the normal-form bidding game defined below.

\paragraph{Budget Balance.}
In the prior-free sections, we use \emph{budget balance} to mean ex-post weak budget balance: for every realized bid profile and every round,
\[
    \sum_{i\in B_t} p_i^t(\bids^t,\sbids^t)
    \ge
    \sum_{j\in S_t} q_j^t(\bids^t,\sbids^t).
\]
In the Bayesian section, where payments are sometimes balanced only in expectation over the prior, we explicitly use the term \emph{ex-ante weak budget balance}.

\paragraph{Efficiency Measures.}
We measure efficiency by \emph{liquid welfare}, the total target-normalized value held by all participants under the final allocation. 
In the autobidding literature, liquid welfare is used instead of raw social welfare because raw social welfare cannot generally be approximated under RoS constraints~\cite[\S2.2]{aggarwal2024autobidding}. 
Formally,
\[
  \wel(\alloc) 
  =
  \sum_t
  \left[
      \sum_{i\in B_t} \frac{v_i^t}{\tau_i}x_i^t
      +
      \sum_{j\in S_t} \frac{c_j^t}{\tau_j}(1-y_j^t)
  \right].
  \tag{Liquid Welfare}
\]
The \emph{liquid gains from trade} (LGFT) of an allocation is the target-normalized value created by the trades:
\begin{align}\label{eq:gft}
    \gft(\alloc)
    =
    \sum_t\sum_{i\in B_t}\sum_{j\in S_t}
    x_{ij}^t
    \left(
        \frac{v_i^t}{\tau_i}
        -
        \frac{c_j^t}{\tau_j}
    \right).
    \tag{LGFT}
\end{align}
For every fixed instance,
$
    \wel(\alloc)
    =
    \sum_t\sum_{j\in S_t}\frac{c_j^t}{\tau_j}
    +
    \gft(\alloc),
$
so maximizing liquid welfare is equivalent to maximizing LGFT.

Since all targets are public and positive, we work in target-normalized units. 
For a buyer with RoS target \(\tau_i\), set \(\bar v_i^t=v_i^t/\tau_i\); for a value-maximizing seller with RoC parameter \(\tau_j\), set \(\bar c_j^t=c_j^t/\tau_j\). 
For utility-maximizing agents, we take the target parameter to be \(1\). 
Under this change of units, the liquid welfare and LGFT of the original instance are exactly the ordinary welfare and GFT of the normalized instance. 
Thus, unless stated otherwise, we suppress the bars and interpret welfare and GFT as the normalized, equivalently liquid, objectives.\footnote{This normalization does not imply that raw social welfare or raw GFT in the original, unnormalized instance is approximable. The approximation guarantees are for liquid welfare/LGFT, equivalently ordinary welfare/GFT after target normalization; when targets differ across agents, this is a different benchmark from raw social welfare/GFT.}

\paragraph{Prior-Free Normal-Form Game, Equilibria and Price of Anarchy (PoA).}
Following the standard autobidding PoA framework (see \cite{aggarwal2019autobidding}), the prior-free part of the paper is modeled as a normal-form game in which each agent chooses a vector of bids, one bid for each round in which the agent participates. 
The buyers $i$ and seller $j$'s pure strategy is a vector
\[
    \bids_i=(b_i^t)_{t:i\in B_t}\in\mathbb{R}_{\ge0}^{|\{t:i\in B_t\}|} \qquad \text{and} \qquad  \sbids_j=(s_j^t)_{t:j\in S_t}\in\mathbb{R}_{\ge0}^{|\{t:j\in S_t\}|}.
\]
There are no history-dependent strategies in this model: agents do not condition their bid in round \(t\) on bids, allocations, or payments observed in earlier rounds. 
The role of the repeated setting is to aggregate values, payments, and RoS/RoC constraints across multiple market instances.

Given a realized type profile \((\val,\cost)\) and a bid-vector profile \((\bids,\sbids)\), define the aggregate quantities
\[
    V_i(\bids,\sbids)
    :=
    \sum_{t:i\in B_t} v_i^t x_i^t(\bids^t,\sbids^t),
    \qquad
    P_i(\bids,\sbids)
    :=
    \sum_{t:i\in B_t} p_i^t(\bids^t,\sbids^t)
\]
for buyer \(i\), and
\[
    C_j(\bids,\sbids)
    :=
    \sum_{t:j\in S_t} c_j^t y_j^t(\bids^t,\sbids^t),
    \qquad
    Q_j(\bids,\sbids)
    :=
    \sum_{t:j\in S_t} q_j^t(\bids^t,\sbids^t)
\]
for seller \(j\). Here \(V_i\) and \(P_i\) are buyer \(i\)'s total received value and total payment, while \(C_j\) and \(Q_j\) are seller \(j\)'s total traded cost and total revenue.

For utility-maximizers, payoffs are quasi-linear:
\[
    U_i^{\mathrm{util}} = V_i-P_i,
    \qquad
    U_j^{\mathrm{util}} = Q_j-C_j.
\]
For value-maximizers, payoffs are the corresponding scale objective whenever the RoS/RoC constraint is satisfied, and \(-\infty\) otherwise:
\[
    U_i^{\mathrm{val}}
    =
    \begin{cases}
        V_i, & \text{if } V_i-\tau_i P_i\ge 0,\\
        -\infty, & \text{otherwise,}
    \end{cases}
    \qquad
    U_j^{\mathrm{val}}
    =
    \begin{cases}
        C_j, & \text{if } \tau_j Q_j-C_j\ge 0,\\
        -\infty, & \text{otherwise.}
    \end{cases}
\]

\paragraph{Hybrid buyers.}
Following the hybrid objective in \cite[\S2.1]{aggarwal2024autobidding}, before target normalization buyer $i$ maximizes $\widetilde V_i-\lambda_iP_i$ subject to $\widetilde V_i\ge\tau_iP_i$, where $\widetilde V_i$ is total unnormalized received value, $\lambda_i\in[0,1]$, and $\tau_i\ge1$ is her public RoS target. Dividing the objective and constraint by $\tau_i$ gives
\begin{equation}\label{eq:hybrid_objective}
    \max_{\bids_i}\ V_i-\beta_iP_i
    \quad\text{subject to}\quad V_i\ge P_i,
    \qquad V_i=\widetilde V_i/\tau_i,
    \qquad \beta_i:=\lambda_i/\tau_i\in[0,1].
\end{equation}
Thus $\beta_i=0$ corresponds to the value-maximizing endpoint and $\beta_i=1$ corresponds to quasi-linear utility maximization.

A pure Nash equilibrium is a profile of bid vectors such that no single agent can improve her payoff by changing her entire bid vector while holding all other agents' bid vectors fixed.
 
We evaluate mechanisms at pure Nash equilibria of the normal-form bidding game defined above. 
This solution concept is standard in the autobidding PoA literature~\cite{aggarwal2019autobidding}.%

For an objective \(\Phi\in\{\gft,\wel\}\), let \(\alloc_\Phi^*(\val,\cost)\) be an allocation maximizing \(\Phi\) under the true values and costs. 
Let \(E_M(\val,\cost)\) be the set of pure Nash equilibrium allocations of mechanism \(M\) on that profile. 
For randomized mechanisms, \(\Phi(\alloc)\) denotes the expected objective value induced by the allocation rule. 
The price of anarchy of \(M\) with respect to \(\Phi\) is
\[
  \poa_\Phi(M)
  =
  \sup_{(\val,\cost): E_M(\val,\cost)\neq \emptyset}
  \frac{
      \Phi(\alloc_\Phi^*(\val,\cost))
  }{
      \inf_{\alloc\in E_M(\val,\cost)} \Phi(\alloc)
  }.
\]
We write \(\gftpoa(M)\) and \(\welpoa(M)\) for the corresponding PoA values.

\section{Unbounded Price of Anarchy for GFT}
\label{sec:prior_free_unbounded_poa}

As discussed in the introduction, a central question in autobidding is whether mechanisms designed for quasi-linear utility-maximizers (e.g., second-price auction) continue to perform well when some agents are value-maximizers; see, e.g.,~\cite{aggarwal2019autobidding,mehta2022auction,deng2021towards,aggarwal2024autobidding}. In the one-sided setting, the second-price auction has a Price of Anarchy of~2 with RoS value-maximizing buyers~\cite{aggarwal2019autobidding}.

We ask the analogous question for two-sided markets. In particular, do classical truthful and budget-balanced double-auction mechanisms, such as Trade Reduction~\citep{MCAFEE}, provide any guarantee for liquid gains from trade when at least one side of the market is value-maximizing? In this section we show that the answer is negative: classical mechanisms cannot guarantee any finite approximation to LGFT.

Throughout this section, a prior-free instance is a finite collection of rounds, and an agent's pure strategy is a vector of bids, one bid for each round in which the agent participates. Thus the repeated market is treated as a normal-form game over bid vectors, as in the standard autobidding PoA model. Each round is one double-auction instance.

Our first result applies to deterministic mechanisms. We consider mechanisms that are truthful and individually rational for utility-maximizers, ex-post weakly budget-balanced, anonymous, and monotone.

\label{app:prior_free_unbounded_poa}
\label{app:poa_gft_lowerbound_truthful}
We first formalize the anonymity and monotonicity assumptions used in \cref{thm:poa_gft_lowerbound_truthful}. Since this theorem concerns deterministic mechanisms, the allocation variables below are indicators.

\begin{definition}[Anonymity]
\label{def:anonymous}
A mechanism $M$ is anonymous if relabeling buyers or sellers relabels the allocation in the same way. Formally, let
\[
    x_i(\bids,\sbids)=\sum_j x_{ij}(\bids,\sbids),
    \qquad
    y_j(\bids,\sbids)=\sum_i x_{ij}(\bids,\sbids)
\]
denote whether buyer $i$ trades and whether seller $j$ trades. For any two buyers $i,i'$, if $\bids'$ is obtained from $\bids$ by swapping $b_i$ and $b_{i'}$, then
\[
    x_i(\bids,\sbids)=x_{i'}(\bids',\sbids).
\]
For any two sellers $j,j'$, if $\sbids'$ is obtained from $\sbids$ by swapping $s_j$ and $s_{j'}$, then
\[
    y_j(\bids,\sbids)=y_{j'}(\bids,\sbids').
\]
\end{definition}

\begin{definition}[Monotonicity]
\label{def:monotone}
A mechanism $M$ is monotone if buyers become weakly more likely to trade when they raise their own bids, and sellers become weakly more likely to trade when they lower their own asks, holding all other reports fixed. Formally, for every profile $(\bids,\sbids)$ and every $\delta>0$,
\begin{align*}
    x_i(\bids,\sbids)
    &\ge x_i(\bids-\delta e_i,\sbids)
    &&\text{for every buyer }i,\\
    y_j(\bids,\sbids)
    &\ge y_j(\bids,\sbids+\delta e_j)
    &&\text{for every seller }j,
\end{align*}
whenever the displayed reports lie in the corresponding bid or ask space. Here $e_i$ and $e_j$ denote the corresponding unit vectors.
\end{definition}

We use the following consequence of utility truthfulness.

\begin{lemma}[Deviations by untraded agents]\label{lem:losing-thresholds}
Fix a deterministic uIC mechanism with zero transfers to nontrading agents, and fix the other agents' reports.
If a buyer does not trade at bid $b_i$, any report that makes her trade charges her at least $b_i$.
If a seller does not trade at ask $s_j$, any report that makes her trade pays her at most $s_j$.
\end{lemma}

\begin{proof}
Consider a utility-maximizing buyer whose true value is $b_i$. Truthful reporting gives allocation zero and payment zero, hence utility zero. If a deviation wins at payment $p_i$, utility truthfulness implies $0\ge b_i-p_i$, so $p_i\ge b_i$.
Similarly, a seller with true cost $s_j$ obtains utility zero from the losing truthful report. Any winning deviation with revenue $q_j$ must satisfy $0\ge q_j-s_j$, hence $q_j\le s_j$.
\end{proof}

\begin{theorem}\label{thm:poa_gft_lowerbound_truthful}
Let $M$ be any deterministic, anonymous, monotone, ex-post WBB, IR, and uIC mechanism{, with zero transfers to nontrading agents, applied identically in every round}. If at least one side of the market consists of value-maximizing agents, then the Price of Anarchy for Liquid Gains from Trade is unbounded: $\gftpoa(M)=\infty$.
Moreover, the bad-equilibrium instances can be chosen so that Trade Reduction under truthful target-normalized reports obtains a constant-factor approximation to the first-best LGFT.
\end{theorem}

The proof constructs an explicit pure Nash equilibrium whose LGFT tends to zero while the optimal LGFT remains bounded away from zero. The key point is that a value-maximizing agent can pool the RoS or RoC constraint across rounds: a profitable trade in one round can subsidize an unprofitable trade in another round, and a utility-truthful mechanism cannot distinguish these rounds once the same bids are submitted.

\begin{proof}[Proof of \cref{thm:poa_gft_lowerbound_truthful}]
We prove the theorem for the case in which buyers are value maximizers. {The seller-side construction is given at the end of the proof.}

For $0<\eps<1/16$, put $\eta=\eps/4$ and consider a single round with two buyers and two sellers submitting the following reports:
\[
\begin{array}{c|cc}
  \textbf{Agent ID} & \textbf{Seller Ask} & \textbf{Buyer Bid}\\
  \hline
  1 & 1-\eps   & 1+\eps\\
  2 & 1-\eps/2 & 1+\eps/2
\end{array}.
\]
Let $k(\eps)\in\{0,1,2\}$ denote the number of trades at this report profile, and let $B$ and $S$ denote the sets of traded buyers and sellers, respectively. Thus $|B|=|S|=k(\eps)$.

\paragraph{Case 1: no trade near the efficient crossing.}

Suppose that for arbitrarily small $\eps>0$ we have $k(\eps)=0$. Fix such an $\eps$. Consider the one-round instance with the reports above and buyer values $v_i=b_i-\eta$. Set $c_j=s_j$ for utility-maximizing sellers and $c_j=s_j+\eta$ for value-maximizing sellers.

By construction, the mechanism makes no trade, so the equilibrium LGFT is zero. {Every buyer value is at least $1+\eps/4$, and every seller cost is at most $1-\eps/4$. Thus both trades have positive LGFT, and truthful Trade Reduction obtains at least half the first-best LGFT.}

It remains to verify that the displayed bid profile is a pure Nash equilibrium. {By \cref{lem:losing-thresholds}, any deviation that makes buyer $i$ trade charges her at least $b_i>v_i$, violating RoS. Similarly, any deviation that makes a value-maximizing seller $j$ trade pays her at most $s_j<c_j$, violating RoC. Utility-maximizing sellers report their true costs and have no profitable deviation by uIC.} Hence no agent has a profitable feasible deviation.

\paragraph{Case 2: at least one trade near the efficient crossing.}

We may therefore assume that for all sufficiently small $\eps>0$, $k(\eps)\ge1$. Fix such an $\eps$ and consider two rounds with the same bids and asks as above. Set
\[
(v_i^1,v_i^2)=
\begin{cases}
    (3/2,\;1/2+2\eps), & i\in B,\\
    (b_i-\eta,\;b_i-\eta), & i\notin B.
\end{cases}
\]
In both rounds, set $c_j^t=s_j$ for traded sellers and for utility-maximizing sellers, and set $c_j^t=s_j+\eta$ for untraded value-maximizing sellers.

Since the same deterministic mechanism is applied to the same reports in both rounds, the traded sets are $B$ and $S$ in both rounds.

Therefore, the equilibrium LGFT is at most

\[
    k(\eps)(2+2\eps)-2\sum_{j\in S}s_j
    \le 4k(\eps)\eps
    \le 8\eps.
\]
The first-best allocation can trade a buyer with value $3/2$ in round~1 with a seller whose cost is at most $1$, obtaining LGFT at least $1/2$.

Thus the ratio between first-best LGFT and equilibrium LGFT is at least {$1/(16\eps)$}, which diverges as $\eps\to0$.

We verify that the displayed bids form a pure Nash equilibrium. {Every buyer in $B$ wins in both rounds.} By ex-post IR for utility-maximizing buyers, her payment in each round is at most her bid {$b_i\le1+\eps$}, so her total payment is at most $2+2\eps$, exactly her total value. Hence her RoS constraint is feasible. {Since she already trades in both rounds, no feasible deviation increases her value.}

Every buyer outside $B$ has value $b_i-\eta$ in each round. By \cref{lem:losing-thresholds}, any newly won round costs at least $b_i$. Thus every nonempty deviation violates her aggregate RoS constraint.

A traded value-maximizing seller has cost $s_j$ in each round and receives at least $s_j$ by ex-post IR. Hence her RoC constraint holds, and she already trades in both rounds. An untraded value-maximizing seller has cost $s_j+\eta$ in each round, while any newly traded round pays her at most $s_j$ by \cref{lem:losing-thresholds}. Thus every nonempty deviation violates RoC. Utility-maximizing sellers report their true costs and have no profitable deviation by uIC.

Thus the profile is a pure Nash equilibrium.

Finally, in round~1 every buyer value exceeds every seller cost, and at least one buyer has value $3/2$. Thus TR with truthful target-normalized reports obtains LGFT at least $1/2$. First-best LGFT across both rounds is at most
\[
    2(1/2+\eps)+4\eps=1+6\eps<2.
\]
Hence truthful TR obtains at least a quarter of first-best LGFT, establishing the claimed constant-factor benchmark.

For the case in which sellers are value maximizers, use the same report profiles. If $k(\eps)=0$ for arbitrarily small $\eps$, fix such an $\eps$ and set $c_j=s_j+\eta$ for every seller, $v_i=b_i$ for utility-maximizing buyers, and $v_i=b_i-\eta$ for value-maximizing buyers. The same argument gives a pure equilibrium with zero LGFT, while truthful TR obtains at least half the positive first-best LGFT.

Otherwise, fix a sufficiently small $\eps$ with $k(\eps)\ge1$, use two rounds, and set
\[
(c_j^1,c_j^2)=
\begin{cases}
    (3/2,\;1/2-2\eps), & j\in S,\\
    (s_j+\eta,\;s_j+\eta), & j\notin S.
\end{cases}
\]
Set $v_i^t=b_i$ for traded buyers and utility-maximizing buyers, and $v_i^t=b_i-\eta$ for untraded value-maximizing buyers. A traded seller receives at least $2s_j\ge2-2\eps$, her total traded cost, and already trades in both rounds. Untraded sellers cannot make a nonempty RoC-feasible deviation by \cref{lem:losing-thresholds}. Traded value-maximizing buyers satisfy RoS by ex-post IR and already trade in both rounds; untraded value-maximizing buyers cannot make a nonempty RoS-feasible deviation. Utility-maximizing buyers report truthfully. Thus this is again a pure Nash equilibrium.

Its LGFT is at most
\[
    2\sum_{i\in B}b_i-k(\eps)(2-2\eps)
    \le4k(\eps)\eps\le8\eps.
\]
In round~2 every buyer value exceeds every seller cost, and at least one seller has cost $1/2-2\eps$. Consequently, both first-best LGFT and truthful-TR LGFT are at least $1/2$. Total first-best LGFT is at most
\[
    4\eps+2(1/2+3\eps)=1+10\eps<2.
\]
Truthful TR therefore again obtains at least a quarter of first-best LGFT, while the equilibrium ratio diverges as $\eps\to0$.

\end{proof}

{The second impossibility removes the determinism and anonymity assumptions, but applies to the case where exactly one side is value-maximizing.}

\begin{theorem}\label{thm:poa_gft_lowerbound_truthful_randomized}
Let $M$ be any uIC, IR, and ex-post WBB mechanism. If sellers are utility-maximizers and buyers are value-maximizers, or symmetrically if buyers are utility-maximizers and sellers are value-maximizers, then the Price of Anarchy for Liquid Gains from Trade is unbounded: $\gftpoa(M)=\infty$.
Moreover, the bad-equilibrium instances can be chosen so that Trade Reduction under truthful target-normalized reports achieves the first-best LGFT.
\end{theorem}

\label{app:poa_gft_lowerbound_truthful_randomized}

\begin{proof}[Proof of \cref{thm:poa_gft_lowerbound_truthful_randomized}]
We prove the case of value-maximizing buyers and utility-maximizing sellers. The reverse case follows by exchanging buyers and sellers and replacing RoS constraints by RoC constraints.

Consider two rounds, each with two buyers and two sellers. All sellers have cost~$1$. Buyer~2 has value~$1$ in both rounds. Buyer~1 has values $v_1$ and $v_2$, which will be chosen below.
\[
\begin{array}{c|cc}
  \multicolumn{3}{c}{\textbf{Round 1}} \\
  \textbf{Agent ID} & \textbf{Seller Cost} & \textbf{Buyer Value} \\
  \hline
  1& 1 & v_1\\
  2& 1 & 1
\end{array}
\qquad
\begin{array}{c|cc}
  \multicolumn{3}{c}{\textbf{Round 2}} \\
  \textbf{Agent ID} & \textbf{Seller Cost} & \textbf{Buyer Value} \\
  \hline
  1& 1 & v_2\\
  2& 1 & 1
\end{array}
\]
Since sellers are utility-maximizers and the mechanism is uIC and IR, the sellers truthfully ask~$1$. Buyer~2 bidding~$1$ in both rounds is a best response: any additional allocation would require crossing a buyer threshold strictly above~$1$, and the associated threshold payment would exceed buyer~2's value.

Now fix the other agents' reports as above and consider buyer~1. Because the two rounds are identical from the mechanism's perspective except for buyer~1's value, buyer~1 faces the same interim allocation and payment rule in both rounds. Let $x(b)$ and $p(b)$ denote buyer~1's allocation probability and expected payment when she bids $b$. Utility-IC and IR imply that $x$ is nondecreasing. We also have $x(b)=0$ for all $b<1$. Indeed, seller IR implies that total seller payments are at least the total allocation probability, buyer~2's IR at value~$1$ implies that buyer~2's payment is at most her allocation probability, and WBB then implies $p(b)\ge x(b)$ for buyer~1. If $b<1$ and $x(b)>0$, a utility-maximizing buyer with value $b$ would have negative utility, contradicting IR. Taking the standard right-continuous representative of the monotone allocation rule, Myerson's payment identity gives
\[
    p(b)=b\,x(b)-\int_1^b x(z)\,dz,\qquad b\ge1.
\]
Let
$
    b_0=\sup\{b:x(b)=0\}\ge1.
$
We distinguish three cases.

\paragraph{Case 1: $b_0>1$.}
Set $v_1=v_2=1+(b_0-1)/2$. Any positive allocation for buyer~1 requires payment at least~$b_0$, which is larger than her value. Hence buyer~1 cannot trade in any RoS-feasible deviation, so the equilibrium LGFT is zero. The first-best LGFT is $b_0-1>0$, obtained by trading buyer~1 in both rounds. This gives an infinite ratio.

\paragraph{Case 2: $b_0=1$ and $x$ is continuous at $1$.}
Choose $v_1=1+\delta$ and $v_2=0$, where $\delta>0$ will be chosen below. The first-best LGFT is $\delta$. Since $x$ is continuous at~$1$ and $x(1)=0$, for every $K>1$ there exists $\eps>0$ such that $x(1+\eps)<1/K$. Choose $\delta\in(0,\eps)$ small enough so that
\[
    \delta < \int_1^{1+\eps}\left(1-\frac{x(z)}{x(1+\eps)}\right)dz,
\]
where if $x(1+\eps)=0$ the conclusion is immediate.

If buyer~1 bids $b_1$ in round~1, her RoS constraint requires
\begin{equation}\label{eq:roi_1_updated}
    x(b_1)(1+\delta)\ge p(b_1)=b_1x(b_1)-\int_1^{b_1}x(z)\,dz .
\end{equation}
Equivalently,
\begin{equation}\label{eq:roi_2_updated}
    \int_{1+\delta}^{b_1}\bigl(x(b_1)-x(z)\bigr)dz
    \le
    \int_1^{1+\delta}x(z)\,dz .
\end{equation}
The choice of $\delta$ implies that \cref{eq:roi_1_updated} fails at $b_1=1+\eps$; since the left-hand side of \cref{eq:roi_2_updated} is nondecreasing in $b_1$, every RoS-feasible bid satisfies $b_1<1+\eps$. Therefore buyer~1's equilibrium allocation probability in round~1 is at most $x(1+\eps)<1/K$, and the equilibrium LGFT is at most $\delta/K$. Since the first-best LGFT is $\delta$, the ratio is at least~$K$.

\paragraph{Case 3: $b_0=1$ and $x$ has a jump at $1$.}
Let $\alpha=x(1)>0$. By right-continuity, for every $K>1$ there exists $\eps>0$ such that
\[
    x(1+\eps)<\alpha+\min\left\{\frac{1}{K},\frac{\alpha}{2}\right\}.
\]
Choose $v_1=1+\delta$ and $v_2=1-\delta$, for a sufficiently small $\delta\in(0,\eps)$ satisfying
\[
    \delta < \min\left\{
    \int_1^{1+\eps}\left(1-\frac{x(z)}{x(1+\eps)}\right)dz,
    \frac{1}{3}
    \right\}.
\]
The first-best LGFT is $\delta$.

Let $b_1,b_2$ be buyer~1's bids in the two rounds. We first claim that in any equilibrium $x(b_2)>0$. If $x(b_2)=0$, then the same RoS calculation as in Case~2 implies $b_1<1+\eps$, and hence
\[
    (1+\delta)x(b_1)<(1+\delta)\cdot\frac{3}{2}\alpha<2\alpha.
\]
By instead bidding $b_1=b_2=1$, buyer~1 obtains allocation probability $\alpha$ in both rounds, total value $\alpha(v_1+v_2)=2\alpha$, and total payment $2\alpha$, satisfying RoS. This is a profitable feasible deviation, contradiction. Therefore $x(b_2)>0$, and by the definition of the jump, $x(b_2)\ge\alpha$.

Since $v_2<1$, any positive allocation in round~2 consumes RoS slack:
\[
    x(b_2)(1-\delta)
    <
    p(b_2)
    =b_2x(b_2)-\int_1^{b_2}x(z)\,dz
    \qquad \text{for every } b_2>1.
\]
Consequently, the high-value round must itself satisfy the RoS inequality from \cref{eq:roi_1_updated}; the argument from Case~2 gives $b_1\le1+\eps$ and therefore
\[
    x(b_1)\le x(1+\eps)<\alpha+\min\left\{\frac{1}{K},\frac{\alpha}{2}\right\}.
\]
The equilibrium LGFT is thus at most
\[
    \delta x(b_1)-\delta x(b_2)
    \le
    \delta\left(\alpha+\min\left\{\frac{1}{K},\frac{\alpha}{2}\right\}\right)-\delta\alpha
    \le
    \frac{\delta}{K}.
\]
Since the first-best LGFT is $\delta$, the ratio is at least~$K$.

Because $K$ is arbitrary, the PoA is unbounded. In the constructed instances, TR with truthful target-normalized reports obtains the first-best LGFT: it trades buyer~1 in the positive-LGFT round and makes no positive-LGFT trade in the other round.
\end{proof}

\begin{remark}
The lower bounds above do not rely on the absence of equilibria. In each constructed instance, the proof explicitly specifies a pure bidding profile and verifies that it is a Nash equilibrium for the relevant value-maximizing and utility-maximizing agents. Thus the PoA ratios are well-defined on the instances used in the lower bounds.
\end{remark}

\begin{remark}
For both \cref{thm:poa_gft_lowerbound_truthful,thm:poa_gft_lowerbound_truthful_randomized}, the lower-bound instances are not pathological instances on which Trade Reduction is inherently poor under truthful reporting. In the deterministic construction, truthful Trade Reduction obtains a constant-factor approximation to first-best LGFT, and in the randomized construction it obtains first-best LGFT. The failure comes from the equilibrium behavior induced by value-maximizers' global RoS/RoC constraints, not from the absence of good trades.
\end{remark}

Our results leave open the possibility that non-uIC mechanisms tailored specifically to autobidding could achieve a finite PoA bound for Liquid GFT in the prior-free setting.

\section{Welfare Price of Anarchy of Trade Reduction}
\label{sec:welfare_poa_tr}
In Section~\ref{sec:prior_free_unbounded_poa}, we proved that no ``classic'' prior-free mechanism can hope to approximate the LGFT.
In this section, we show that classic prior-free mechanisms \emph{can} approximate the liquid welfare.
In fact, we make use of the simple Trade Reduction mechanism due to \cite{MCAFEE} which is defined as follows.
For a double auction instance, let $b_1 \geq \ldots \geq b_m$ be the bids of the buyers and $s_1 \leq \ldots \leq s_n$ denote the bids of the sellers.
If $b_1 < s_1$, there is no trade so assume $b_1 \geq s_1$.
Let $k^* = \max\{ k \,:\, b_k \geq s_k \}$.
We then trade buyers $b_1, \ldots, b_{k^*-1}$ and sellers $s_1, \ldots, s_{k^*-1}$.
Traded buyers pay $b_{k^*}$ while traded sellers are paid $s_{k^*}$.
In other words, we compute the optimal matching, \emph{reduce} the last trade, and use the reduced buyer's bid and seller's ask as the prices for the traded buyers and traded sellers, respectively.

In this section, we allow buyers to be hybrid buyers. See \cref{sec:prelim} for the formal definition. Let $\beta$ be defined in \cref{eq:hybrid_objective}. The endpoint $\beta_i=0$ gives a value-maximizing buyer, while $\beta_i=1$ gives the quasi-linear objective. For TR, a utility-maximizing buyer has a best response consisting only of rounds with nonnegative utility, so the aggregate RoS constraint does not restrict the optimum at this endpoint.

Before we state our results, we require the following observation.
\begin{observation}
For deterministic uIC threshold mechanisms such as Trade Reduction:
\begin{itemize}[itemsep=0pt,topsep=0pt]
\item A {hybrid} buyer has no incentive to bid below its value.
\item A value-maximizing seller has no incentive to bid above its cost.
\end{itemize}
\end{observation}
  
To see this, suppose a buyer bids below its value.  
If the buyer wins the round, then increasing its bid (up to its true value) leaves the allocation and price unchanged, so there is no loss in doing so.  
If the buyer does not win the round, there are two cases to consider.  
If raising the bid to its value still does not result in winning, then underbidding provides no benefit.  
If, on the other hand, increasing the bid to its value would result in winning, then the buyer is incentivized to do so, {since the new trade has threshold price $\pi\le v_i^t$, contributes $v_i^t-\beta_i\pi\ge0$ to her hybrid objective, and does not decrease her RoS slack.}

A similar argument applies to show that a seller has no incentive to bid above its cost.
Throughout this section, following the standard no-dominated-bids convention, we restrict attention to bid profiles in which {hybrid buyers do not bid below their target-normalized values} and value-maximizing sellers do not ask above their target-normalized costs.

Utility-maximizing agents report truthfully. All upper bounds below are for feasible pure equilibria satisfying these conventions. Write $\underline\beta=\min_i\beta_i$ and $r=\min_t|\tbo(t)|\ge2$, where each optimal allocation has maximum cardinality among welfare-optimal allocations, with fixed tie-breaking.

Our first result concerns hybrid buyers and utility-maximizing sellers.
\begin{theorem}[Hybrid buyers and utility-maximizing sellers]
\label{thm:one_side_value_poa}
\label{thm:welfare_poa_value_buyers_utility_sellers}
Consider a sequence of double auctions with hybrid buyers \cref{eq:hybrid_objective} and utility-maximizing sellers. Let $r\ge2$ denote the size of the smallest optimal matching among the rounds. Then
\[
    \welpoa(TR)\le 2-\underline\beta+\frac{1}{r-1}.
\]
If all buyers share a common parameter $\beta_i=\beta$, then
\[
    \welpoa(TR)=2-\beta+\frac{1}{r-1}.
\]
\end{theorem}

Our second result concerns hybrid buyers and value-maximizing sellers.
\begin{theorem}[Hybrid buyers and value-maximizing sellers]
\label{thm:two_side_value_poa}
\label{thm:welfare_poa_two_side_upperbound}
Consider a sequence of double auctions with hybrid buyers \cref{eq:hybrid_objective} and value-maximizing sellers. Let $r\ge2$ denote the size of the smallest optimal matching among the rounds. Then
\[
    \welpoa(TR)\le 3-\underline\beta+\frac{1}{r-1}.
\]
If all buyers share a common parameter $\beta_i=\beta$, then
\[
    3-\beta-\frac{2-\beta}{r+1}
    \le\welpoa(TR)\le
    3-\beta+\frac{1}{r-1}.
\]
\end{theorem}

At $\beta=0$, these theorems recover the value-maximizing buyer bounds: $2+1/(r-1)$ with utility-maximizing sellers, and the interval $[3-2/(r+1),\,3+1/(r-1)]$ with value-maximizing sellers. At $\beta=1$, the first theorem recovers the truthful utility-maximizer ratio $1+1/(r-1)$, while the second gives
\[
    2-\frac{1}{r+1}\le\welpoa(TR)\le2+\frac{1}{r-1}
\]
for utility-maximizing buyers and value-maximizing sellers. The lower-bound construction for this endpoint can use truthful buyer reports, as verified below.

We first give a high-level overview of the proofs, followed by the full arguments.
We focus on the setting where buyers are hybrid buyers and sellers are utility-maximizing as an example, though all cases follow a similar approach.

For each round \(t\), let \(\tbo(t)\) denote the set of buyers traded in the optimal allocation, and \(\tbv(t)\) the set of buyers traded at equilibrium (as in Theorems~\ref{thm:one_side_value_poa} and~\ref{thm:two_side_value_poa}). Similarly, let \(\rso(t)\) and \(\rsv(t)\) denote the sets of rejected sellers in the optimal and equilibrium allocations, respectively. It is straightforward that the equilibrium two-sided welfare is at least the welfare of \(\tbv(t) \cup \rsv(t)\), and hence at least the welfare of \((\tbv(t) \cap \tbo(t)) \cup (\rsv(t) \cap \rso(t))\).

What remains is to account for the welfare lost from \(\tbo(t) \setminus \tbv(t)\) and \(\rso(t) \setminus \rsv(t)\).
To do so, we “charge” this welfare loss to the payments made by traded buyers. More concretely, buyers in \(\tbv(t)\) can be partitioned into three groups:
\begin{enumerate}
  \item Those also in \(\tbo(t)\), whose welfare is directly accounted for.
  \item Those who replace buyers in \(\tbo(t) \setminus \tbv(t)\). These buyers “kick out” the missing optimal buyers, so the missing buyers\textquotesingle{} bids lower-bound the price paid by every traded buyer, covering the welfare loss from \(\tbo(t) \setminus \tbv(t)\).
  \item Those who create new trades that force sales with sellers in \(\rso(t) \setminus \rsv(t)\). Since sellers are truthful and utility-maximizing, individual rationality ensures they receive at least their cost, covering the welfare loss from these sellers.
\end{enumerate}

This partition is not perfectly clean, which introduces an additive \(\frac{1}{r-1}\) term in the Price of Anarchy bound.

Finally, using the Return-on-Spend constraint, total buyer welfare upper bounds total buyer payments. At the value-maximizing endpoint $\beta_i=0$, combining these observations, we conclude that roughly twice the equilibrium welfare (plus a small additive term) covers the welfare of \(\tbo(t) \cup \rso(t)\).

For a hybrid buyer $i$ traded at buyer price $p^t$, equilibrium additionally implies $v_i^t\ge\beta_i p^t$: otherwise dropping the trade improves her objective and relaxes RoS. Thus every equilibrium buyer outside the optimum already contributes at least $\underline\beta p^t$ to welfare. Subtracting this contribution leaves only $(1-\underline\beta)p^t$ to charge for each such buyer. This yields $\mathrm{OPT}\le W+(1-\underline\beta+1/(r-1))P$ with utility-maximizing sellers, where $W$ and $P$ are total equilibrium welfare and buyer payments. Since $P\le W$, the first bound follows. With value-maximizing sellers, their global RoC constraints and budget balance cover the missing retained-seller costs by one additional copy of $W$.

\begin{remark}
The hybrid upper bounds also apply to heterogeneous buyer parameters and a mix of utility-maximizing and value-maximizing sellers, using the value-seller upper bound in the latter case. Utility-maximizing sellers report truthfully and satisfy ex-post IR; value-maximizing sellers satisfy the global RoC constraint and the no-overasking convention. A buyer population mixing the two original objectives is the special case $\beta_i\in\{0,1\}$.
\end{remark}

\subsection{Welfare bounds and proofs}
\label{app:welfare_poa_tr}
\label{app:one_side_value_poa}
\label{app:two_side_value_poa}
\paragraph{Useful notation.}
For a round \(t\), let \(\tbu(t)\), \(\tbv(t)\), and \(\tbo(t)\) denote, respectively, the set of buyers traded by TR under truthful target-normalized reports, the set of buyers traded by TR at the equilibrium bidding profile, and the set of buyers traded by an optimal welfare-maximizing allocation.
Similarly, let \(\rsu(t)\), \(\rsv(t)\), and \(\rso(t)\) denote the corresponding sets of rejected sellers.
For a buyer set \(A\) and seller set \(B\), write
\[
    v_t(A)=\sum_{i\in A} v_i^t,
    \qquad
    c_t(B)=\sum_{j\in B} c_j^t .
\]
Let \(\TwoSidedWelfare(t)\), \(\BuyerWelfare(t)\), \(\SellerWelfare(t)\), \(\BuyerRevenue(t)\), and \(\SellerRevenue(t)\) denote the total welfare, buyer welfare, seller retained-cost welfare, total buyer payment, and total seller revenue achieved by TR at equilibrium in round \(t\).
The notation \(\BuyerRevenue\) is kept for consistency with the main text, but it should be read as buyer payment.

Let \(r_t=|\tbo(t)|\) and \(r=\min_t r_t\ge2\).
Under truthful target-normalized reports, TR trades all but the marginal optimal trade, so
\[
    |\tbu(t)|=r_t-1,
    \qquad
    \tbu(t)\subseteq \tbo(t),
    \qquad
    \rso(t)\subseteq\rsu(t).
\]

Let $S_t$ denote the set of sellers present in round $t$. The subscript $v$ continues to denote the equilibrium allocation, now allowing hybrid buyers. Utility-maximizing agents report truthfully, hybrid buyers do not underbid, and value-maximizing sellers do not overask. All welfare and payment quantities use target-normalized units.

\begin{proof}[Proof of the upper bounds in Theorems~\ref{thm:one_side_value_poa} and~\ref{thm:two_side_value_poa}]
Fix a pure equilibrium satisfying the stated conventions. Write $p^t=p_b^t$ for the TR buyer price, $q_t=|\tbv(t)|$, $W_t=\TwoSidedWelfare(t)$, and $P_t=q_tp^t$; set $W=\sum_tW_t$ and $P=\sum_tP_t$. By buyer RoS,
\begin{equation}\label{eq:hybrid-payment-value}
P\le\sum_t\BuyerWelfare(t)\le W.
\end{equation}
Increasing bids and decreasing asks relative to truthful reports cannot decrease the reported efficient trade count. Therefore $q_t\ge r_t-1\ge r-1$.
Every untraded buyer has value at most $p^t$, since her bid is at least her value.
Every traded hybrid buyer satisfies
\begin{equation}\label{eq:hybrid-local}
v_i^t\ge\beta_i p^t\ge\underline\beta p^t.
\end{equation}
Indeed, if $v_i^t<\beta_i p^t$, dropping just this trade strictly increases her hybrid objective. It also relaxes RoS, because $\beta_i\le1$ implies $v_i^t<p^t$. This is a feasible unilateral deviation, contradicting equilibrium.

Let $a_t=|\tbo(t)\cap\tbv(t)|$, $d_t=|\tbo(t)\setminus\tbv(t)|=r_t-a_t$, and $e_t=|\tbv(t)\setminus\tbo(t)|=q_t-a_t$.
By \eqref{eq:hybrid-local}, the equilibrium buyers outside the optimum contribute at least $\underline\beta e_tp^t$ to welfare.

\paragraph{Utility-maximizing sellers.}
Truthful sellers trade in increasing order of cost. Hence
$\ell_t:=|\rso(t)\setminus\rsv(t)|=\max\{q_t-r_t,0\}$, and each such seller has cost at most $p_s^t\le p^t$. Consequently,
\[
d_t+\ell_t=e_t+\max\{r_t-q_t,0\}\le e_t+1.
\]
Subtracting the equilibrium-only buyers' welfare before charging the missing optimal terms gives
\begin{align*}
v_t(\tbo(t))+c_t(\rso(t))
&\le W_t+\bigl(d_t+\ell_t-\underline\beta e_t\bigr)p^t\\
&\le W_t+\bigl((1-\underline\beta)e_t+1\bigr)p^t\\
&\le W_t+\left(1-\underline\beta+\frac1{r-1}\right)P_t.
\end{align*}
Summing over rounds and applying \eqref{eq:hybrid-payment-value},
\[
\mathrm{OPT}\le W+\left(1-\underline\beta+\frac1{r-1}\right)P
\le\left(2-\underline\beta+\frac1{r-1}\right)W.
\]

\paragraph{Value-maximizing sellers.}
Here $d_t=e_t+r_t-q_t\le e_t+1$. The same buyer charging argument yields
\[
v_t(\tbo(t))+c_t(\rsv(t))
\le W_t+\left(1-\underline\beta+\frac1{r-1}\right)P_t.
\]
The costs of optimal retained sellers that trade at equilibrium are covered using seller RoC and ex-post weak budget balance:
\[
\sum_t c_t(\rso(t)\setminus\rsv(t))
\le\sum_t c_t(S_t\setminus\rsv(t))
\le\sum_t\SellerRevenue(t)\le P\le W.
\]
Combining the last two inequalities proves
\[
\mathrm{OPT}\le 2W+\left(1-\underline\beta+\frac1{r-1}\right)P
\le\left(3-\underline\beta+\frac1{r-1}\right)W.
\]
Taking $\beta_i=1$ gives the upper bound for utility-maximizing buyers and value-maximizing sellers. Taking $\beta_i=0$ gives the original value-buyer bounds.
\end{proof}

\begin{lemma}\label{lem:welfare_poa_one_side_lowerbound}
For every common hybrid parameter $\beta\in[0,1]$ and integer $r\ge2$, the welfare PoA of TR with hybrid buyers and utility-maximizing sellers is at least $2-\beta+1/(r-1)$.
\end{lemma}
\begin{proof}
First suppose $\beta<1$, and choose $0<\varepsilon<(1-\beta)/4$.
Consider a two-round instance with $r$ sellers and $2r-1$ buyers in each round. All sellers have cost $\varepsilon$ and bid truthfully.
Round 1 is:
\[
\begin{array}{c|c|cc}
\textbf{Agent ID}&\textbf{Seller Cost}&\textbf{Buyer Value}&\textbf{Buyer Bid}\\\hline
1,\ldots,r&\varepsilon&1&1\\
r+1,\ldots,2r-1&-&\beta+2\varepsilon&1+\varepsilon
\end{array}
\]
Round 2 is:
\[
\begin{array}{c|c|cc}
\textbf{Agent ID}&\textbf{Seller Cost}&\textbf{Buyer Value}&\textbf{Buyer Bid}\\\hline
1,\ldots,r&\varepsilon&2\varepsilon&2\varepsilon\\
r+1,\ldots,2r-1&-&1-\beta&1-\beta
\end{array}
\]
TR trades buyers $r+1,\ldots,2r-1$ with $r-1$ sellers in both rounds, at buyer prices $1$ and $2\varepsilon$, respectively. These buyers each receive total value $1+2\varepsilon$ and pay $1+2\varepsilon$, so RoS binds. Their objective contributions in the two rounds are $2\varepsilon$ and $1-\beta-2\beta\varepsilon$, both nonnegative. They already win both rounds, and dropping either round cannot improve their hybrid objective.

For each remaining buyer, winning round 1 requires payment $1+\varepsilon>1$, and winning round 2 requires payment $1-\beta>2\varepsilon$. Neither round can provide RoS slack, so no nonempty deviation is feasible. Sellers report truthfully. Thus the displayed bids form a pure equilibrium with no underbidding.

Both rounds have optimal trade count $r$. The equilibrium welfare and optimal welfare are, respectively,
\[
W=(r-1)+2r\varepsilon,
\qquad
\mathrm{OPT}=r+(r-1)(1-\beta)+2\varepsilon.
\]
Letting $\varepsilon\to0$ gives $2-\beta+1/(r-1)$.
At $\beta=1$, take a single round with $r$ buyers of value $1$ and $r$ sellers of cost $0$, all reporting truthfully. TR welfare is $r-1$, whereas optimal welfare is $r$, giving $1+1/(r-1)$, as required.
\end{proof}

\begin{lemma}\label{lem:welfare_poa_two_side_lowerbound}
For every common hybrid parameter $\beta\in[0,1]$ and integer $r\ge2$, the welfare PoA of TR with hybrid buyers and value-maximizing sellers is at least $3-\beta-(2-\beta)/(r+1)$.
\end{lemma}
\begin{proof}
Consider a two-round instance. Choose $0<\varepsilon<1/6$ and, when $\beta<1$, also $\varepsilon<(1-\beta)/4$.
Round 1 has buyers
\[
\begin{array}{c|cc}
\textbf{Buyer Agent ID}&\textbf{Value}&\textbf{Bid}\\\hline
1,\ldots,r&\beta+\varepsilon&1+\varepsilon\\
r+1,\ldots,2r&1&1
\end{array}
\]
and sellers
\[
\begin{array}{c|cc}
\textbf{Seller Agent ID}&\textbf{Cost}&\textbf{Bid}\\\hline
1,\ldots,r&0&0\\
r+1&1+\varepsilon&1-\varepsilon
\end{array}.
\]
Round 2 has buyers
\[
\begin{array}{c|cc}
\textbf{Buyer Agent ID}&\textbf{Value}&\textbf{Bid}\\\hline
1,\ldots,r&1-\beta+3\varepsilon&M\\
r+1&3\varepsilon&3\varepsilon\\
r+2,\ldots,2r&0&0
\end{array}
\]
where $M$ exceeds the total value in the instance, and sellers
\[
\begin{array}{c|cc}
\textbf{Seller Agent ID}&\textbf{Cost}&\textbf{Bid}\\\hline
1,\ldots,r&1&\varepsilon\\
r+1,\ldots,2r&2\varepsilon&2\varepsilon
\end{array}.
\]
Sellers $r+2,\ldots,2r$ appear only in round 2.

TR trades buyers $1,\ldots,r$ with sellers $1,\ldots,r$ in both rounds. In round 1, buyers pay $1$ and sellers receive $1-\varepsilon$. In round 2, buyers pay $3\varepsilon$ and sellers receive $2\varepsilon$.
Each traded buyer receives total value $1+4\varepsilon$ and pays $1+3\varepsilon$. Her hybrid objective contributions are $\varepsilon$ and $(1-\beta)(1+3\varepsilon)$, both nonnegative. Each traded seller receives revenue $1+\varepsilon$ against total cost $1$. These buyers and sellers already trade in both rounds, so no feasible deviation improves their objectives.

The other buyers cannot win round 1 without paying $1+\varepsilon>1$. In round 2 they would pay $M$, exceeding their total possible value, so they cannot generate offsetting RoS slack. An untraded seller who wins round 1 receives $0$, and one who wins round 2 receives $\varepsilon<2\varepsilon$; no such seller has a RoC-feasible nonempty deviation. Thus the displayed profile is a pure equilibrium under the stated conventions.

The optimal trade count is $r$ in each round, including when $\beta=1$. The equilibrium welfare is
\[
W=r(\beta+\varepsilon)+(1+\varepsilon)
+r(1-\beta+3\varepsilon)+2r\varepsilon
=r+1+(6r+1)\varepsilon.
\]
A feasible allocation trades buyers $r+1,\ldots,2r$ with sellers $1,\ldots,r$ in round 1, and buyers $1,\ldots,r$ with sellers $r+1,\ldots,2r$ in round 2. Its welfare is
\[
r+(1+\varepsilon)+r(1-\beta+3\varepsilon)+r
=(3-\beta)r+1+(3r+1)\varepsilon.
\]
Letting $\varepsilon\to0$ gives
\[
\frac{(3-\beta)r+1}{r+1}
=3-\beta-\frac{2-\beta}{r+1}.
\]
\end{proof}

\begin{corollary}\label{lem:welfare_poa_utility_buyers_value_sellers_lower_asymptotic}
With utility-maximizing buyers and value-maximizing sellers, the welfare ratio can approach $2-1/(r+1)$ for each $r\ge2$, and hence approaches $2$ as $r\to\infty$.
\end{corollary}
\begin{proof}
Take $\beta=1$ in the construction of \cref{lem:welfare_poa_two_side_lowerbound}. The buyers' selected rounds have nonnegative quasi-linear utility. Replacing their winning bids by truthful values preserves the allocation and prices with the fixed tie-breaking priority favoring buyers $1,\ldots,r$ in round 2. The seller equilibrium and welfare calculations are unchanged.
\end{proof}

\begin{corollary}
\label{thm:welfare_poa_utility_buyers_value_sellers}
Under the conventions above, every pure equilibrium of TR with utility-maximizing buyers and value-maximizing sellers has welfare PoA at most \(2+1/(r-1)\), where \(r=\min_t|\tbo(t)|\).
\end{corollary}

\begin{proof}
Fix a pure equilibrium.
Utility-maximizing buyers bid truthfully. Since value-maximizing sellers do not overask under the no-dominated-bids convention, the number of TR trades is at least the number of truthful-TR trades. Therefore the buyers served by truthful TR, \(\tbu(t)\), are also served at equilibrium. Hence, for every round \(t\),
\begin{equation}
\label{eq:uv_tsw_basic}
\TwoSidedWelfare(t)
\ge
v_t(\tbu(t))+c_t(\rsv(t)).
\end{equation}
Since \(\tbu(t)\) consists of the top \(r_t-1\) optimal buyers and \(\tbo(t)\) consists of the top \(r_t\) optimal buyers,
\[
    \frac{r_t}{r_t-1}v_t(\tbu(t))\ge v_t(\tbo(t)).
\]
Because \(r_t\ge r\), we also have \(r/(r-1)\ge r_t/(r_t-1)\). Multiplying \cref{eq:uv_tsw_basic} by \(r/(r-1)\) gives
\begin{equation}
\label{eq:uv_tsw_scaled}
\left(1+\frac{1}{r-1}\right)\TwoSidedWelfare(t)
\ge
v_t(\tbo(t))+c_t(\rsv(t)).
\end{equation}

Next, summing across rounds, buyer IR and ex-post budget balance give
\[
    \sum_t\BuyerWelfare(t)
    \ge
    \sum_t\BuyerRevenue(t)
    \ge
    \sum_t\SellerRevenue(t).
\]
The value-maximizing sellers' global RoC constraints imply
\[
    \sum_t\SellerRevenue(t)
    \ge
    \sum_t c_t({S_t}\setminus\rsv(t))
    \ge
    \sum_t c_t(\rsu(t)\setminus\rsv(t)).
\]
Therefore,
\begin{equation}
\label{eq:uv_seller_charge}
    \sum_t\TwoSidedWelfare(t)
    \ge
    \sum_t\BuyerWelfare(t)
    \ge
    \sum_t c_t(\rsu(t)\setminus\rsv(t)).
\end{equation}
Combining \cref{eq:uv_tsw_scaled,eq:uv_seller_charge}, and using \(\rso(t)\subseteq\rsu(t)\), yields
\begin{align*}
\left(2+\frac{1}{r-1}\right)\sum_t\TwoSidedWelfare(t)
&\ge
\sum_t \Bigl(v_t(\tbo(t))+c_t(\rsv(t))+c_t(\rsu(t)\setminus\rsv(t))\Bigr)\\
&\ge
\sum_t \Bigl(v_t(\tbo(t))+c_t(\rso(t))\Bigr).
\end{align*}
The final expression is the optimal welfare.
\end{proof}
\subsection{Price of Anarchy with Respect to Truthful TR Welfare}
\label{app:poa_tr}

In this subsection, the benchmark is the welfare obtained by TR when all agents report their true target-normalized values and costs.
Equivalently, the numerator is
\[
    \sum_t \bigl(v_t(\tbu(t))+c_t(\rsu(t))\bigr)
\]
instead of first-best welfare.

\paragraph{Value-maximizing buyers and utility-maximizing sellers.}

\begin{lemma}
\label{lem:welfare_poa_one_side_lowerbound_tr}
The PoA of TR with respect to truthful-TR welfare is at least \(2\) with value-maximizing buyers and utility-maximizing sellers.
\end{lemma}

\begin{proof}
Take two sellers with cost zero in every round.
TR then reduces exactly one trade and is equivalent to a single-item second-price auction among the buyers.
The lower bound of~\cite{aggarwal2019autobidding} for second-price auctions with value-maximizing bidders therefore gives instances whose ratio approaches \(2\).
\end{proof}

\begin{theorem}
\label{thm:welfare_poa_one_side_tr}
The PoA of TR with respect to truthful-TR welfare is at most \(2\) with value-maximizing buyers and utility-maximizing sellers.
\end{theorem}

\begin{proof}
For every round \(t\),
\[
\TwoSidedWelfare(t)
\ge
v_t(\tbu(t)\cap\tbv(t))+c_t(\rsv(t)).
\]
Moreover, by the same price and counting argument as in \cref{thm:welfare_poa_value_buyers_utility_sellers}, but with \(\tbu,\rsu\) in place of \(\tbo,\rso\), we have
\[
\BuyerRevenue(t)
\ge
v_t(\tbu(t)\setminus\tbv(t))+c_t(\rsu(t)\setminus\rsv(t)).
\]
Summing over rounds and using buyer RoS,
\[
    \sum_t\BuyerRevenue(t)\le\sum_t\BuyerWelfare(t)\le\sum_t\TwoSidedWelfare(t),
\]
gives
$
    2\sum_t\TwoSidedWelfare(t)
    \ge
    \sum_t \bigl(v_t(\tbu(t))+c_t(\rsu(t))\bigr).$
\end{proof}

\paragraph{Utility-maximizing buyers and value-maximizing sellers.}

\begin{lemma}
\label{lem:welfare_tr_poa_lb_utility_buyers_value_sellers}
The PoA of TR with respect to truthful-TR welfare is at least \(2\) with utility-maximizing buyers and value-maximizing sellers.
\end{lemma}

\begin{proof}
The $\beta=1$ case in \cref{lem:welfare_poa_two_side_lowerbound} has equilibrium welfare \((r+1)+O(r\varepsilon)\), while truthful TR obtains welfare at least \(2r-O(r\varepsilon)\).
Letting \(\varepsilon\to0\) and then \(r\to\infty\) gives a ratio approaching \(2\).
\end{proof}

\begin{theorem}
\label{thm:welfare_tr_poa_utility_buyers_value_sellers}
The PoA of TR with respect to truthful-TR welfare is at most \(2\) with utility-maximizing buyers and value-maximizing sellers.
\end{theorem}

\begin{proof}
For every round \(t\),
\[
\TwoSidedWelfare(t)
\ge
v_t(\tbu(t))+c_t(\rsv(t)),
\]
because buyers bid truthfully and value-maximizing sellers do not overask under the no-dominated-bids convention.
As in \cref{thm:welfare_poa_utility_buyers_value_sellers}, buyer IR, budget balance, and seller RoC imply
\[
    \sum_t\TwoSidedWelfare(t)
    \ge
    \sum_t c_t(\rsu(t)\setminus\rsv(t)).
\]
Therefore,
\begin{align*}
2\sum_t\TwoSidedWelfare(t)
&\ge
\sum_t\bigl(v_t(\tbu(t))+c_t(\rsv(t))+c_t(\rsu(t)\setminus\rsv(t))\bigr)\\
&\ge
\sum_t\bigl(v_t(\tbu(t))+c_t(\rsu(t))\bigr). \qedhere
\end{align*}
\end{proof}

\paragraph{Value-maximizing buyers and value-maximizing sellers.}

\begin{lemma}
\label{lem:welfare_poa_two_side_lowerbound_tr}
The PoA of TR with respect to truthful-TR welfare is at least \(3\) with value-maximizing buyers and value-maximizing sellers.
\end{lemma}

\begin{proof}
A variant of the lower-bound construction in \cref{lem:welfare_poa_two_side_lowerbound} gives ratio approaching \(3\) against truthful-TR welfare.
Use the same two rounds, but compare the equilibrium welfare to the TR welfare under truthful target-normalized reports.
The equilibrium welfare is \(r+1+O(r\varepsilon)\), while truthful TR obtains welfare at least \(3r-1-O(r\varepsilon)\).
Letting \(\varepsilon\to0\) and \(r\to\infty\) gives the claim.
\end{proof}

\begin{theorem}
\label{thm:welfare_poa_two_side_tr}
The PoA of TR with respect to truthful-TR welfare is at most \(3\) with value-maximizing buyers and value-maximizing sellers.
\end{theorem}

\begin{proof}
For every round \(t\), the equilibrium welfare covers
\[
    v_t(\tbu(t)\cap\tbv(t))+c_t(\rsv(t)).
\]
Because value-maximizing buyers do not underbid under the no-dominated-bids convention, every buyer in \(\tbu(t)\setminus\tbv(t)\) has value at most the equilibrium buyer price. Since \(|\tbu(t)\setminus\tbv(t)|\le|\tbv(t)|\), we have
\[
    \BuyerRevenue(t)
    \ge
    v_t(\tbu(t)\setminus\tbv(t)).
\]
Buyer RoS, budget balance, and seller RoC imply
\[
    \sum_t\BuyerWelfare(t)
    \ge
    \sum_t c_t(\rsu(t)\setminus\rsv(t)).
\]
Hence
\begin{align*}
3\sum_t\TwoSidedWelfare(t)
&\ge
\sum_t\Bigl(\BuyerWelfare(t)+\SellerWelfare(t)+\BuyerRevenue(t)+\BuyerWelfare(t)\Bigr)\\
&\ge
\sum_t\Bigl(v_t(\tbu(t))+c_t(\rsv(t))+c_t(\rsu(t)\setminus\rsv(t))\Bigr)\\
&\ge
\sum_t\Bigl(v_t(\tbu(t))+c_t(\rsu(t))\Bigr),
\end{align*}
which is truthful-TR welfare.
\end{proof}

\begin{corollary}[Hybrid buyers and the truthful-TR benchmark]\label{cor:hybrid-tr-benchmark}
With hybrid buyers, the PoA relative to truthful-TR liquid welfare is at most $2-\underline\beta$ with utility-maximizing sellers and at most $3-\underline\beta$ with value-maximizing sellers. For a common parameter $\beta_i=\beta$, both bounds are tight.
\end{corollary}
\begin{proof}
Repeat the upper-bound proof using $\tbu(t),\rsu(t)$ in place of $\tbo(t),\rso(t)$. Put $k_t=|\tbu(t)|=r_t-1$. Since $q_t\ge k_t$, the missing and extra buyer counts satisfy $d_t\le e_t$. With truthful sellers, their missing retained-seller count is $\ell_t=q_t-k_t$, so $d_t+\ell_t=e_t$. Thus there is no marginal-trade additive term, giving $\mathrm{TR}_{\mathrm{truth}}\le W+(1-\underline\beta)P$ with utility sellers and $\mathrm{TR}_{\mathrm{truth}}\le2W+(1-\underline\beta)P$ with value sellers. Use $P\le W$.

For utility sellers and $\beta<1$, the construction in \cref{lem:welfare_poa_one_side_lowerbound} has truthful-TR welfare tending to $(r-1)(2-\beta)$ and equilibrium welfare tending to $r-1$ as $\varepsilon\to0$. The endpoint $\beta=1$ is immediate from truthful reporting. For value sellers, the construction in \cref{lem:welfare_poa_two_side_lowerbound} has truthful-TR welfare at least $(3-\beta)r+\beta-1-O(r\varepsilon)$ and equilibrium welfare $r+1+O(r\varepsilon)$. Let $\varepsilon\to0$ and then $r\to\infty$.
\end{proof}

\subsection{Existence of Equilibria for Trade Reduction}
\label{app:equilibriaexistance}

In this subsection, we prove the existence of pure Nash equilibria satisfying the no-dominated-bids convention used in Section~\ref{sec:welfare_poa_tr} for hybrid buyers and utility-maximizing sellers, and for the reverse one-sided value-maximizer case.
Recall that the repeated market is modeled as a normal-form game: each agent submits a fixed vector of bids, one bid for each round in which the agent participates.
We use ``round'' throughout; this is the same object sometimes called a ``query'' in online-advertising terminology.

We fix a deterministic tie-breaking rule for Trade Reduction.
Equivalently, one may perturb equal bids and asks by infinitesimals according to a fixed priority order.
The proof below only uses the resulting consistency of the reduced buyer and reduced seller.

Let
\[
    M > \sum_t \sum_{i\in B_t} v_i^t + 1
\]
be a bid larger than the total value any {hybrid} buyer can obtain across all rounds.

\begin{theorem}
\label{thm:existence_value_buyers_utility_sellers}
Suppose buyers are {hybrid buyers}, sellers are utility-maximizers, and every round is run by Trade Reduction with the fixed tie-breaking rule described above.
Then the repeated bidding game admits a pure Nash equilibrium satisfying the no-dominated-bids convention.
\end{theorem}

\begin{proof}
Since TR is DSIC and ex-post IR for quasi-linear utility-maximizing sellers in each round, truthful bidding is a dominant strategy for every seller.
We therefore fix seller bids to be their true costs and construct a pure Nash equilibrium for the buyers satisfying the no-underbidding convention.

Initialize every buyer's bid in every round to her true value and run TR in every round.
For every buyer-round pair in which the buyer trades, raise that buyer's bid in that round to \(M\).
This does not affect the allocation or prices in that round: the buyer-side price is the bid of the reduced buyer, and the reduced buyer is not one of the traded buyers whose bid was raised to \(M\).

We now process the buyers one at a time in a fixed order.
We maintain the invariant that, in every round, every currently traded {hybrid} buyer bids \(M\), while every currently untraded {hybrid} buyer bids her true value.
Thus the no-underbidding convention is maintained throughout.

Given the bids of all other buyers and the truthful seller bids, buyer \(i\) faces a posted-price problem across rounds.
For each round \(t\), define \(\pi_i^t\in \mathbb{R}_{\ge0}\cup\{\infty\}\) to be the payment buyer \(i\) would make in round \(t\) if she changed only her bid in that round to \(M\) and traded.
If she cannot trade at any finite payment, set \(\pi_i^t=\infty\).
Under TR and the invariant above, this price is independent of the exact bid as long as the bid is large enough to make buyer \(i\) trade.

Buyer \(i\)'s best-response problem is therefore
\[
    \max_{S\subseteq\{t:i\in B_t\}} \sum_{t\in S} {(v_i^t-\beta_i\pi_i^t)}
    \qquad
    \text{s.t.}
    \qquad
    \sum_{t\in S} \pi_i^t \le \sum_{t\in S} v_i^t .
\]
Among optimal solutions, choose one containing all rounds currently won by buyer \(i\).
{Such a choice exists because, before buyer $i$ is processed, every currently won round has price $\pi_i^t\le v_i^t$. Keeping it weakly relaxes RoS and adds $v_i^t-\beta_i\pi_i^t\ge0$ to the objective. The same argument lets us choose an optimal set containing every other round with $\pi_i^t\le v_i^t$ that is available at her truthful bid; hence bidding her value outside the chosen set does not force an unwanted trade.}
Buyer \(i\) then bids \(M\) in the rounds in this chosen set and bids her value in all other rounds.

It remains to verify that processing buyer \(i\) does not destroy the best responses of previously processed buyers.
Consider any round.
If buyer \(i\) already traded in the round, nothing changes.
If buyer \(i\) was the reduced buyer, then bidding \(M\) would either still leave her reduced or would make the buyer-side price equal to \(M\).
Since \(M\) exceeds buyer \(i\)'s total possible value, such a round is never part of a feasible best response.
Finally, suppose buyer \(i\) was neither traded nor reduced.
If bidding \(M\) lets her trade at a finite price, then the efficient matching with respect to bids contains all previously traded buyers, buyer \(i\), and the previously reduced buyer.
The previously reduced buyer still has the lowest bid among these buyers and remains the reduced buyer.
Hence all previously traded buyers continue to trade at the same price, while buyer \(i\) is added as a new traded buyer.
If buyer \(i\) cannot trade at a finite price, then nothing changes except that the posted price for some untraded buyers may increase to \(\infty\).

Thus after buyer \(i\)'s update:
\begin{enumerate}[itemsep=0pt,topsep=0pt]
    \item buyer \(i\)'s winning set only expands;
    \item the winning set and prices of every previously processed buyer are unchanged;
    \item prices on rounds not won by a previously processed buyer weakly increase.
\end{enumerate}
Therefore every previously processed buyer’s strategy remains a best response.
After all buyers have been processed, every buyer is best-responding and every seller is truthfully reporting.
The resulting bid profile is a pure Nash equilibrium satisfying the no-dominated-bids convention.
\end{proof}

\begin{corollary}
\label{cor:existence_utility_buyers_value_sellers}
The symmetric one-sided case also admits a pure Nash equilibrium satisfying the no-dominated-bids convention: if buyers are utility-maximizers, sellers are value-maximizers, and every round is run by TR with the same fixed tie-breaking convention, then such an equilibrium exists.
\end{corollary}

\begin{proof}[Proof sketch]
The proof is the seller-side analogue of Theorem~\ref{thm:existence_value_buyers_utility_sellers}.
Utility-maximizing buyers bid truthfully because TR is DSIC and ex-post IR for quasi-linear buyers.
For value-maximizing sellers, initialize asks at true costs, lower the ask of every currently traded value-maximizing seller to the lowest ask allowed by the bid space, breaking ties by the fixed priority rule, and process sellers one by one.
The no-overasking convention is maintained throughout.

Given the other agents' reports, each seller faces a posted-revenue problem across rounds: lowering her ask enough either makes her trade at a revenue determined by the reduced seller, or cannot make her trade at a finite feasible revenue.
Each seller chooses a revenue-feasible set of rounds maximizing total traded cost subject to the RoC constraint.
{The analogous argument using TR\textquotesingle{}s allocation and payment rules shows} that when a seller adds a round, previously processed sellers keep their winning sets and revenues, while revenues on rounds not chosen by previously processed sellers can only become less attractive.
Hence a single pass over sellers yields a pure Nash equilibrium satisfying the no-dominated-bids convention.
\end{proof}

\begin{remark}
\label{rem:existence_two_sided_value}
The argument above is specific to the one-sided cases, where one side is quasi-linear and therefore truthfully reports under TR.
{It includes heterogeneous hybrid buyers, and it extends to a mixed seller population in the reverse one-sided case by fixing the utility-maximizing sellers' truthful reports and processing only the value-maximizing sellers.}
The proof does not, by itself, establish pure-equilibrium existence when both buyers and sellers are value-maximizers: then both sides' bid changes can affect the posted prices/revenues faced by the other side, and the monotone one-pass invariant used above need not survive.
The welfare upper bounds in Section~\ref{sec:welfare_poa_tr} should therefore be read as applying to pure equilibria under the no-dominated-bids convention that exist; the lower-bound constructions explicitly exhibit such equilibria.
\end{remark}


\section{Full GFT Mechanism for Bayesian Matching Markets}
\label{sec:fullgft}

In this section, we turn to the Bayesian setting, where the designer has access to prior distributional information about agents' private values and costs. Our goal is to design mechanisms tailored to Return-on-Spend (RoS) autobidders. We show that, by leveraging the value-maximizing behavior of these agents, it is possible to restore full Liquid Gains from Trade (LGFT) while satisfying truthfulness, individual rationality (IR), and ex-ante weak budget balance (WBB).

Our mechanism applies to general matching markets with downward-closed feasibility constraints. This stands in sharp contrast to the impossibility result of Myerson and Satterthwaite, which established that no mechanism can simultaneously satisfy BIC, IR, WBB, and first-best efficiency for utility maximizers, even in the simplest bilateral trade setting \cite{MS83}.

\begin{theorem*}[Informal]
Given any matching market instance with value-maximizing buyers and utility-maximizing sellers, \cref{mech:onedimension} achieves the optimal LGFT and liquid welfare ex post. It is truthful for value-maximizing buyers, DSIC for utility-maximizing sellers, IR, and ex-ante WBB.
\end{theorem*}

This contrasts with \cref{sec:prior_free_unbounded_poa}, where we show that classical mechanisms can have unbounded Price of Anarchy. By tailoring the payment rule to RoS behavior, the mechanism circumvents the Myerson--Satterthwaite impossibility, which applies to quasi-linear utility maximizers.

\begin{theorem*}[\cite{MS83}]
When all agents are utility maximizers, no BIC, IR, and WBB mechanism can guarantee full efficiency, even in a bilateral trade instance with one buyer and one seller.
\end{theorem*}

All the missing proofs can be found in ~\cref{app:seller_truthful}.

\subsection{Model Setup}

We consider a single-round matching market with $n$ buyers and $m$ sellers. Buyer $i$'s value and seller $j$'s cost are independently drawn from publicly known distributions, denoted by $v_i\sim D_i$ and $c_j\sim S_j$. A mechanism takes buyers' bids $\bids=(b_1,\dots,b_n)$ and sellers' bids $\sbids=(s_1,\dots,s_m)$ as input and outputs an allocation $\alloc$ and payments $\pb,\ps$. Here, $x_{ij}$ denotes whether buyer $i$ trades with seller $j$, $p_i^b$ denotes the payment collected from buyer $i$, and $p_j^s$ denotes the payment paid to seller $j$.

Let $E=\{(i,j):i\in[n],j\in[m]\}$ be the set of feasible buyer-seller pairs. We use $\mathcal F\subseteq 2^E$ to denote the feasible sets of trades that can be implemented simultaneously. We assume that every feasible set is a matching and that $\mathcal F$ is downward closed: if $A\in\mathcal F$ and $A'\subseteq A$, then $A'\in\mathcal F$. This includes the standard double-auction setting as a special case. We also fix a deterministic tie-breaking rule for maximum-weight feasible matchings.

For a value-maximizing buyer, the objective is to maximize expected received value subject to an expected RoS constraint. Since we use normalized units, the RoS target is one. Let $O_i$ denote buyer $i$'s optimal expected value:
\begin{align}
    O_i
    ={}&\max_f\ \E_{v_i,\val_{-i},\cost}\left[v_i\sum_j x_{ij}(f(v_i),\val_{-i},\cost)\right] \nonumber\\
    &\text{s.t. }\E_{v_i,\val_{-i},\cost}\left[v_i\sum_j x_{ij}(f(v_i),\val_{-i},\cost)-p_i^b(f(v_i),\val_{-i},\cost)\right]\ge 0 . \tag{RoS}
\end{align}
Here $f:\R_{\ge0}\to\Delta(\R_{\ge0})$ is an arbitrary randomized bidding strategy; we do not restrict attention to uniform bid scalings. A mechanism is vBIC for value-maximizing buyers if, for every buyer $i$,
\begin{align}\label{eq:buyer-vbic}
    O_i=
    \E_{\val,\cost}\left[v_i\sum_j x_{ij}(v_i,\val_{-i},\cost)\right]. \tag{vBIC-B}
\end{align}

For utility-maximizing agents, we use the standard notions of DSIC, BIC, IR, and ex-ante WBB from the two-sided-market literature. Formal definitions are stated in \cref{app:prelimBays}.

\begin{remark}[Conditional versus global truthfulness for value maximizers]
For utility-maximizers, conditioning on the other side's realized reports gives a stronger IC requirement: a threshold-paid seller is DSIC profile by profile, and hence also BIC after averaging. For value maximizers the comparison is different because IC is defined through a constrained optimization problem. If we fix a seller-cost state \(s\) and require the RoS constraint to hold conditional on that same \(s\), then we shrink the set of feasible deviations relative to the global definition, where RoS is required only in expectation over both buyers' values and sellers' costs. Thus conditioning on \(s\) can make value-maximizer truthfulness easier, not harder. 
\end{remark}

\subsection{The Mechanism}

At a high level, the mechanism computes the first-best normalized-welfare matching. Utility-maximizing sellers are paid their procurement-side threshold payments, which gives DSIC for those sellers. Value-maximizing buyers are charged their buyer-side threshold payments scaled up by an agent-specific ex-ante constant so that truthful bidding makes each buyer's global RoS constraint bind.

For a fixed allocation rule $\alloc^*$, let $q_i^B(\bids,\sbids)$ denote the buyer-side Myerson threshold payment for buyer $i$, and let $q_j^S(\bids,\sbids)$ denote the procurement-side Myerson threshold payment for seller $j$. We write $x_i(\bids,\sbids)=\sum_jx_{ij}(\bids,\sbids)$ and $x_j(\bids,\sbids)=\sum_ix_{ij}(\bids,\sbids)$.

\begin{mechanism}\label{mech:onedimension}
Given bids \((\bids,\sbids)\), assign weight \(w_{ij}(b_i,s_j)=b_i-s_j\) to each feasible edge \((i,j)\in E\).
\begin{enumerate}
    \item Compute a maximum-weight feasible matching  using the fixed tie-breaking rule.
    \[
        \alloc^*(\bids,\sbids)
        \in
        \argmax_{\alloc\in\mathcal F}
        \sum_{i,j}x_{ij}(b_i-s_j),
    \]

    \item For every utility-maximizing seller \(j\), set the seller payment to the procurement-side threshold payment:
    \[
        p_j^s(\bids,\sbids)=q_j^S(\bids,\sbids).
    \]

    \item For every value-maximizing buyer \(i\), set
    \begin{equation}\label{eq:buyerpayment}
        p_i^b(\bids,\sbids)
        =
        \frac{1+\zeta_i}{\zeta_i}\,q_i^B(\bids,\sbids),
    \end{equation}
    where \(q_i^B\) is the buyer-side threshold payment and \(\zeta_i>0\) is the ex-ante constant chosen so that truthful reporting makes buyer \(i\)'s RoS constraint bind:
    \begin{equation}\label{eq:buyerexpectedpayment}
       \E_{\val,\cost}\!\left[v_i x_i^*(\val,\cost)\right]
       =
       \E_{\val,\cost}\!\left[p_i^b(\val,\cost)\right].
    \end{equation}
\end{enumerate}
\end{mechanism}

\begin{remark}
    \cref{mech:onedimension} is stated for value-maximizing buyers and utility-maximizing sellers. The symmetric one-sided result, with utility-maximizing buyers and value-maximizing sellers, is obtained by applying the same scaling idea to the seller side and using the standard buyer-side threshold payments for utility-maximizing buyers.
\end{remark} 

We first show that the buyer payment rule, together with the efficient allocation, is vBIC-B. The proof uses the payment-scaled threshold-payment argument of \cite{balseiro2021landscape}, adapted to the global expectation over both values and costs.

\begin{lemma}\label{lem:bicb}
For value-maximizing buyers, the allocation rule $\alloc^*$ together with the buyer payments in \cref{eq:buyerpayment} is vBIC-B under the global RoS definition in \cref{eq:buyer-vbic}.
\end{lemma}

\begin{proof}
Fix an active value-maximizing buyer \(i\); inactive buyers with zero expected allocation and payment are trivial. Let \(q_i^B\) be the buyer-side threshold payment for the allocation rule \(\alloc^*\). By \cref{mech:onedimension},
\[
    p_i^b(\bids,\sbids)
    =
    \frac{1+\zeta_i}{\zeta_i}q_i^B(\bids,\sbids),
\]
where \(\zeta_i>0\) is chosen so that truthful reporting makes buyer \(i\)'s RoS constraint bind:
\[
    \E_{\val,\cost}\!\left[p_i^b(\val,\cost)\right]
    =
    \E_{\val,\cost}\!\left[v_i x_i^*(\val,\cost)\right].
\]

Let \(T\) denote truthful reporting, and let \(f\) be any randomized deviation by buyer \(i\). Define
\[
    V_i^f
    :=
    \E_{v_i,\val_{-i},\cost}
    \!\left[v_i x_i^*(f(v_i),\val_{-i},\cost)\right],
    \qquad
    P_i^f
    :=
    \E_{v_i,\val_{-i},\cost}
    \!\left[p_i^b(f(v_i),\val_{-i},\cost)\right].
\]
Define \(V_i^T\) and \(P_i^T\) analogously for truthful reporting.

Since \(q_i^B\) is the buyer-side threshold payment for the monotone allocation rule \(\alloc^*\), Myerson's lemma implies that truthful reporting maximizes quasi-linear utility with respect to the unscaled threshold payment. Hence, for every deterministic deviation, and therefore by linearity for every randomized deviation,
\[
    V_i^T
    -
    \E\!\left[q_i^B(T)\right]
    \ge
    V_i^f
    -
    \E\!\left[q_i^B(f)\right].
\]
Using \(q_i^B=\frac{\zeta_i}{1+\zeta_i}p_i^b\), this becomes
\begin{equation}
\label{eq:buyer_scaled_myerson}
    V_i^T
    -
    \frac{\zeta_i}{1+\zeta_i}P_i^T
    \ge
    V_i^f
    -
    \frac{\zeta_i}{1+\zeta_i}P_i^f .
\end{equation}

By construction, truthful reporting makes the RoS constraint bind, so \(P_i^T=V_i^T\). Therefore the left-hand side of \cref{eq:buyer_scaled_myerson} is
\[
    V_i^T-\frac{\zeta_i}{1+\zeta_i}V_i^T
    =
    \frac{1}{1+\zeta_i}V_i^T.
\]
Now suppose \(f\) is globally RoS-feasible, so \(P_i^f\le V_i^f\). Then the right-hand side of \cref{eq:buyer_scaled_myerson} is at least
\[
    V_i^f-\frac{\zeta_i}{1+\zeta_i}V_i^f
    =
    \frac{1}{1+\zeta_i}V_i^f.
\]
Thus
\[
    \frac{1}{1+\zeta_i}V_i^T
    \ge
    \frac{1}{1+\zeta_i}V_i^f,
\]
and since \(\zeta_i>0\), we get \(V_i^T\ge V_i^f\). Therefore truthful reporting maximizes buyer \(i\)'s expected value among all globally RoS-feasible deviations. This proves vBIC-B.
\end{proof}

We next establish ex-ante budget balance. The key observation is that a utility-maximizing seller's procurement threshold payment is no larger than the value of the buyer to whom she is matched.

\begin{lemma}\label{lem:seller-threshold-bound}
Fix any realized profile $(\val,\cost)$. If seller $j$ is matched to buyer $i$ by $\alloc^*(\val,\cost)$, then her procurement-side threshold payment satisfies
\[
    q_j^S(\val,\cost)\le v_i.
\]
If seller $j$ is unmatched, then $q_j^S(\val,\cost)=0$.
\end{lemma}

\begin{proof}
Fix a profile \((\val,\cost)\), and suppose seller \(j\) is matched to buyer \(i\) in
\(M=\alloc^*(\val,\cost)\). Let
\[
    W
    =
    \sum_{(a,\ell)\in M}(v_a-c_\ell)
\]
be the maximum weight at the truthful profile.

The procurement threshold for seller \(j\) is the highest report at which seller \(j\) can remain allocated. Suppose seller \(j\) raises her reported cost from \(c_j\) to \(s_j=c_j+\Delta\). Every feasible matching that includes seller \(j\) loses exactly \(\Delta\) in weight, while every matching that excludes seller \(j\) is unchanged. Since \(M\) is optimal and includes \(j\), the best matching that includes \(j\) after this deviation has weight at most \(W-\Delta\).

Because \(\mathcal F\) is downward closed, removing edge \((i,j)\) from \(M\) gives a feasible matching that excludes \(j\) and has truthful weight
\[
    W-(v_i-c_j).
\]
This matching is unaffected by seller \(j\)'s report. Hence, if
\[
    \Delta>v_i-c_j,
\]
then some matching excluding \(j\) has strictly larger weight than every matching including \(j\). Equivalently, seller \(j\) cannot remain allocated at any report \(s_j>v_i\). Therefore her procurement-side threshold payment is at most \(v_i\). If \(j\) is unmatched, the threshold payment is zero by the normalization of the procurement payment rule.
\end{proof}

\begin{lemma}\label{lem:bb}
With value-maximizing buyers and utility-maximizing sellers, \cref{mech:onedimension} is ex-ante WBB.
\end{lemma}

\begin{proof}
By \cref{lem:seller-threshold-bound}, for every realized profile $(\val,\cost)$,
\begin{align}\label{eq:seller-payment-upper-bound}
    \sum_j p_j^s(\val,\cost)
    =\sum_j q_j^S(\val,\cost)
    \le \sum_{i,j}x_{ij}^*(\val,\cost)v_i .
\end{align}
Taking expectations and using \cref{eq:buyerexpectedpayment},
\begin{align*}
    \E_{\val,\cost}\left[\sum_i p_i^b(\val,\cost)\right]
    =\E_{\val,\cost}\left[\sum_i v_i x_i^*(\val,\cost)\right]
    =\E_{\val,\cost}\left[\sum_{i,j}x_{ij}^*(\val,\cost)v_i\right]
    \ge \E_{\val,\cost}\left[\sum_j p_j^s(\val,\cost)\right]. \quad \qedhere
\end{align*}
\end{proof}

We are now ready to state the main theorem of this section. 

\begin{theorem}\label{thm:gftoneside}
Given any matching market instance with value-maximizing buyers and utility-maximizing sellers, \cref{mech:onedimension} achieves the optimal LGFT and liquid welfare ex post. It is vBIC-B, BIC-S, IR, and ex-ante WBB.
\end{theorem}

\begin{proof}
The allocation rule in \cref{mech:onedimension} maximizes
\[
    \sum_{i,j}x_{ij}(v_i-c_j)
\]
at every realized profile, so it achieves first-best LGFT ex post. Since the initial seller-endowment term is constant, the same allocation also maximizes liquid welfare ex post.

Utility-maximizing sellers are paid procurement-side threshold payments, so they are DSIC, and hence BIC-S. Seller IR follows from the standard procurement-threshold payment rule. Buyer vBIC-B follows from \cref{lem:bicb}. Buyer ex-ante IR follows from the payment-tightness condition in \cref{eq:buyerexpectedpayment}, which gives
\[
    \E_{\val,\cost}\!\left[v_i x_i^*(\val,\cost)-p_i^b(\val,\cost)\right]=0
\]
for every active value-maximizing buyer \(i\), in normalized units. Finally, ex-ante WBB follows from \cref{lem:bb}.
\end{proof}

\subsection{Value-Maximizing Sellers}

We now extend the construction to markets that may include value-maximizing sellers.
The allocation rule and buyer payment rule remain unchanged.
Utility-maximizing sellers continue to receive procurement-side threshold payments.

For every value-maximizing seller \(j\), set
\begin{equation}\label{eq:value-seller-payment}
    p_j^s(\bids,\sbids)
    =
    \frac{\theta_j}{1+\theta_j}\,q_j^S(\bids,\sbids),
\end{equation}
where \(q_j^S\) is the procurement-side threshold payment and \(\theta_j>0\) is the ex-ante constant chosen so that truthful reporting makes seller \(j\)'s RoC constraint bind:
\begin{equation}\label{eq:value-seller-tightness}
    \E_{\val,\cost}\!\left[p_j^s(\val,\cost)\right]
    =
    \E_{\val,\cost}\!\left[c_jx_j^*(\val,\cost)\right].
\end{equation}

\begin{theorem}\label{thm:bayesian-both-side}
Under the nondegeneracy assumptions above, the extended mechanism is vBIC-B for value-maximizing buyers, vBIC-S for value-maximizing sellers, BIC-S for utility-maximizing sellers, IR, and ex-ante WBB.
If all sellers are value maximizers, then the broker's expected profit equals the expected first-best LGFT.
\end{theorem}

\begin{proof}
Buyer vBIC-B is unchanged from \cref{lem:bicb}. Utility-maximizing sellers, if present, are paid procurement-side threshold payments and are therefore DSIC and IR.

It remains to verify vBIC-S for value-maximizing sellers. Fix an active value-maximizing seller \(j\); inactive sellers with zero expected allocation and payment are trivial. Let \(q_j^S\) be the procurement-side threshold payment for the allocation rule \(\alloc^*\). The extended mechanism sets
\[
    p_j^s(\bids,\sbids)
    =
    \frac{\theta_j}{1+\theta_j}q_j^S(\bids,\sbids),
\]
where \(\theta_j>0\) is chosen so that truthful reporting makes seller \(j\)'s RoC constraint bind:
\[
    \E_{\val,\cost}\!\left[p_j^s(\val,\cost)\right]
    =
    \E_{\val,\cost}\!\left[c_jx_j^*(\val,\cost)\right].
\]

Let \(T\) denote truthful reporting, and let \(f\) be any randomized deviation by seller \(j\). Define
\[
    C_j^f
    :=
    \E_{\val,c_j,\cost_{-j}}
    \!\left[c_jx_j^*(\val,f(c_j),\cost_{-j})\right],
    \qquad
    P_j^f
    :=
    \E_{\val,c_j,\cost_{-j}}
    \!\left[p_j^s(\val,f(c_j),\cost_{-j})\right].
\]
Define \(C_j^T\) and \(P_j^T\) analogously for truthful reporting.

Since \(q_j^S\) is the procurement-side threshold payment for \(\alloc^*\), truthful reporting maximizes quasi-linear seller utility with respect to the unscaled threshold payment. Therefore, for every randomized deviation \(f\),
\[
    \E\!\left[q_j^S(T)\right]-C_j^T
    \ge
    \E\!\left[q_j^S(f)\right]-C_j^f.
\]
Using \(q_j^S=\frac{1+\theta_j}{\theta_j}p_j^s\), this becomes
\begin{equation}
\label{eq:seller_scaled_myerson}
    \frac{1+\theta_j}{\theta_j}P_j^T-C_j^T
    \ge
    \frac{1+\theta_j}{\theta_j}P_j^f-C_j^f.
\end{equation}

By construction, truthful reporting makes seller \(j\)'s RoC constraint bind, so \(P_j^T=C_j^T\). Hence the left-hand side of \cref{eq:seller_scaled_myerson} is
\[
    \frac{1+\theta_j}{\theta_j}C_j^T-C_j^T
    =
    \frac{1}{\theta_j}C_j^T.
\]
Now suppose \(f\) is globally RoC-feasible, so \(P_j^f\ge C_j^f\). Then the right-hand side of \cref{eq:seller_scaled_myerson} is at least
\[
    \frac{1+\theta_j}{\theta_j}C_j^f-C_j^f
    =
    \frac{1}{\theta_j}C_j^f.
\]
Thus
\[
    \frac{1}{\theta_j}C_j^T
    \ge
    \frac{1}{\theta_j}C_j^f,
\]
and since \(\theta_j>0\), we get \(C_j^T\ge C_j^f\). Therefore truthful reporting maximizes seller \(j\)'s expected traded cost among all globally RoC-feasible deviations. This proves vBIC-S.

Ex-ante WBB is preserved because every value-maximizing seller is paid a weakly smaller amount than the corresponding procurement-side threshold payment:
\[
    p_j^s(\bids,\sbids)
    =
    \frac{\theta_j}{1+\theta_j}q_j^S(\bids,\sbids)
    \le
    q_j^S(\bids,\sbids).
\]
Utility-maximizing sellers, if present, receive exactly \(q_j^S\). Thus total expected seller payments in the extended mechanism are no larger than in the one-sided mechanism of \cref{thm:gftoneside}, while buyer payments are unchanged. Since \cref{thm:gftoneside} establishes ex-ante WBB for the one-sided mechanism, the extended mechanism is also ex-ante WBB.

Finally, suppose all sellers are value maximizers. The buyer and seller scaling rules make all RoS and RoC constraints bind in expectation:
\[
    \E_{\val,\cost}\!\left[\sum_i p_i^b(\val,\cost)\right]
    =
    \E_{\val,\cost}\!\left[\sum_{i,j}x_{ij}^*(\val,\cost)v_i\right],
\qquad \text{and} \qquad
    \E_{\val,\cost}\!\left[\sum_j p_j^s(\val,\cost)\right]
    =
    \E_{\val,\cost}\!\left[\sum_{i,j}x_{ij}^*(\val,\cost)c_j\right].
\]
Therefore the broker's expected profit is
\[
    \E_{\val,\cost}\!\left[\sum_i p_i^b-\sum_j p_j^s\right]
    =
    \E_{\val,\cost}\!\left[
        \sum_{i,j}x_{ij}^*(\val,\cost)(v_i-c_j)
    \right],
\]
which is exactly the expected first-best LGFT in normalized units.
\end{proof}

\newpage

\bibliographystyle{alpha}
\bibliography{arxiv_revision/biblio}

\appendix
\crefalias{section}{appendix}
\crefalias{subsection}{appendix}

\section{Motivation for Value-Maximizing Sellers}
\label{app:value_seller_motivation}

A value-maximizing seller is a seller whose primary objective is not short-run quasi-linear profit, but rather the total volume, utilization, or cost-weighted amount of goods sold, subject to a revenue-over-cost constraint. 
This captures sellers who want to expand market share, maintain utilization, or grow a user base while ensuring that revenue covers costs. 
For example, an energy supplier may wish to maximize the amount of electricity sold in an emerging market while satisfying a break-even or revenue-adequacy constraint. 
A cloud or AI-service provider may initially prioritize usage and adoption, while requiring revenue to cover compute and operating costs. 
Similarly, a marketplace seller may temporarily prioritize sales volume or market penetration subject to not selling below an acceptable aggregate revenue-to-cost ratio. 
The RoC constraint plays the same role for sellers as the RoS constraint plays for buyers: it restricts operational feasibility, while the objective captures scale rather than quasi-linear surplus.


\section{Additional Bayesian Preliminaries and Supporting Lemmas}
\subsection{Preliminaries for the Bayesian Setting}\label{app:prelimBays}

\paragraph{Matching market.}
We consider a two-sided matching market with $n$ unit-demand buyers and $m$ unit-supply sellers. Let $v_i\in\mathcal V$ and $c_j\in\mathcal C$ denote buyer $i$'s private value and seller $j$'s private cost for one unit. Values and costs are independently drawn from publicly known distributions $D_i$ and $S_j$. The environment is single-dimensional on both sides.

Let $E=\{(i,j):i\in[n],j\in[m]\}$ be the set of possible trading pairs. A feasibility constraint $\mathcal F\subseteq2^E$ specifies which sets of pairs can be traded simultaneously. We assume that every $A\in\mathcal F$ is a matching and that $\mathcal F$ is downward closed: if $A\in\mathcal F$ and $A'\subseteq A$, then $A'\in\mathcal F$. We write $x_{ij}=1$ if buyer $i$ trades with seller $j$, and zero otherwise.

\paragraph{Mechanism notation.}
A mechanism takes buyers' bids $\bids=(b_1,\dots,b_n)$ and sellers' bids $\sbids=(s_1,\dots,s_m)$ and outputs an allocation $\alloc$ and payments $\pb,\ps$. The payment $p_i^b$ is collected from buyer $i$, and $p_j^s$ is paid to seller $j$. We use
\[
    x_i(\bids,\sbids)=\sum_jx_{ij}(\bids,\sbids),
    \qquad
    x_j(\bids,\sbids)=\sum_ix_{ij}(\bids,\sbids).
\]

\paragraph{Target normalization.}
The main text works in target-normalized units. For a buyer with public RoS target $\tau_i$, replace $v_i$ by $v_i/\tau_i$. For a value-maximizing seller with public RoC target $\tau_j$, replace $c_j$ by $c_j/\tau_j$. Utility-maximizing sellers have no RoC target in the one-sided result and correspond to $\tau_j=1$. After this normalization, all RoS/RoC targets are one and are suppressed.

\paragraph{Utility maximizers.}
A utility-maximizing buyer maximizes value minus payment, and a utility-maximizing seller maximizes payment minus cost. A mechanism is BIC for utility-maximizing sellers, denoted \eqref{eq:seller_bic}, if for every seller $j$ and every report $s_j$,
\begin{align}\label{eq:seller_bic}
    \E_{\val,\cost_{-j}}\left[p_j^s(\val,c_j,\cost_{-j})-c_jx_j(\val,c_j,\cost_{-j})\right]
    \ge
    \E_{\val,\cost_{-j}}\left[p_j^s(\val,s_j,\cost_{-j})-c_jx_j(\val,s_j,\cost_{-j})\right]. \tag{BIC-S}
\end{align}
Analogously, a mechanism is BIC for utility-maximizing buyers, denoted \eqref{eq:buyer_bic}, if for every buyer $i$ and every report $b_i$,
\begin{align}\label{eq:buyer_bic}
    \E_{\val_{-i},\cost}\left[v_ix_i(v_i,\val_{-i},\cost)-p_i^b(v_i,\val_{-i},\cost)\right]
    \ge
    \E_{\val_{-i},\cost}\left[v_ix_i(b_i,\val_{-i},\cost)-p_i^b(b_i,\val_{-i},\cost)\right]. \tag{BIC-B}
\end{align}
A mechanism is ex-ante IR if each agent's expected quasi-linear utility from truthful participation is nonnegative. It is ex-ante WBB if
\begin{align}\label{eq:wbb}
    \E_{\val,\cost}\left[\sum_i p_i^b(\val,\cost)\right]
    \ge
    \E_{\val,\cost}\left[\sum_j p_j^s(\val,\cost)\right]. \tag{WBB}
\end{align}

\paragraph{RoS/RoC value maximizers.}
A value-maximizing buyer maximizes expected received value subject to an expected RoS constraint. In normalized units, buyer $i$ solves
\begin{align*}
    O_i^B
    ={}&\max_f\ \E_{v_i,\val_{-i},\cost}\left[v_ix_i(f(v_i),\val_{-i},\cost)\right]\\
    &\text{s.t. }\E_{v_i,\val_{-i},\cost}\left[v_ix_i(f(v_i),\val_{-i},\cost)-p_i^b(f(v_i),\val_{-i},\cost)\right]\ge0.
\end{align*}
The mechanism is vBIC-B if truthful reporting attains this optimum:
\begin{align}\label{eq:buyer-vbic-app}
    O_i^B=\E_{\val,\cost}\left[v_ix_i(v_i,\val_{-i},\cost)\right]. \tag{vBIC-B}
\end{align}

A value-maximizing seller maximizes expected traded cost subject to an expected RoC constraint. In normalized units, seller $j$ solves
\begin{align*}
    O_j^S
    ={}&\max_f\ \E_{\val,c_j,\cost_{-j}}\left[c_jx_j(\val,f(c_j),\cost_{-j})\right]\\
    &\text{s.t. }\E_{\val,c_j,\cost_{-j}}\left[p_j^s(\val,f(c_j),\cost_{-j})-c_jx_j(\val,f(c_j),\cost_{-j})\right]\ge0.
\end{align*}
The mechanism is vBIC-S if truthful reporting attains this optimum:
\begin{align}\label{eq:seller-abic}
    O_j^S=\E_{\val,\cost}\left[c_jx_j(\val,c_j,\cost_{-j})\right]. \tag{vBIC-S}
\end{align}

\begin{theorem}[Myerson's lemma for buyers and sellers]\label{thm:myerson}
In a single-parameter buyer environment, an allocation rule is implementable for utility-maximizing buyers if and only if it is monotone nondecreasing in each buyer's bid. With the normalization that a zero bid gives zero payment, the threshold payment for buyer $i$ is
\[
    q_i^B(b_i,\bids_{-i},\sbids)
    =b_ix_i(b_i,\bids_{-i},\sbids)-\int_0^{b_i}x_i(t,\bids_{-i},\sbids)\,dt.
\]
In a single-parameter procurement environment, an allocation rule is implementable for utility-maximizing sellers if and only if it is monotone nonincreasing in each seller's reported cost. With the normalization that an unallocated seller receives zero payment, the threshold payment for seller $j$ is
\[
    q_j^S(\bids,s_j,\sbids_{-j})
    =s_jx_j(\bids,s_j,\sbids_{-j})+\int_{s_j}^{\infty}x_j(\bids,t,\sbids_{-j})\,dt,
\]
assuming the integral is finite, as is the case for bounded cost supports.
\end{theorem}

\begin{lemma}[Adapted from \cite{balseiro2021landscape}]\label{lem:yuan}
Under the standard Slater/nondegeneracy condition for the RoS constraint, a mechanism is vBIC-B under the global definition in \cref{eq:buyer-vbic-app} if and only if, for each buyer $i$, there exists a constant $\gamma_i\ge0$ such that:
\begin{enumerate}
    \item for all $v_i,v_i'\in\mathcal V$,
    \begin{align*}
        \E_{\val_{-i},\cost}\left[
        v_ix_i(v_i,\val_{-i},\cost)-\frac{\gamma_i}{1+\gamma_i}p_i^b(v_i,\val_{-i},\cost)
        \right]
        \ge
        \E_{\val_{-i},\cost}\left[
        v_ix_i(v_i',\val_{-i},\cost)-\frac{\gamma_i}{1+\gamma_i}p_i^b(v_i',\val_{-i},\cost)
        \right];
    \end{align*}
    \item the truthful report satisfies the RoS constraint with complementary slackness:
    \begin{align*}
        \E_{\val,\cost}\left[v_ix_i(v_i,\val_{-i},\cost)-p_i^b(v_i,\val_{-i},\cost)\right]\ge0
        \quad\perp\quad
        \gamma_i\ge0.
    \end{align*}
\end{enumerate}
\end{lemma}



\subsection{Truthfulness of Value-Maximizing Sellers}
\label{app:seller_truthful}

The following lemma is the abstract form of the seller-side argument above. It says that if a payment rule can be rescaled so that seller truthfulness is quasi-linear IC and the truthful RoC constraint binds, then truthful reporting is optimal for a value-maximizing seller.

\begin{lemma}[Sufficient condition for vBIC-S]
\label{lem:seller_truthful}
Consider normalized seller RoC constraints. Suppose that, for each value-maximizing seller \(j\), there exists a constant \(\theta_j>0\) such that:
\begin{enumerate}
    \item for all \(c_j,c'_j\in\mathcal C\),
    \[
        \E_{\val,\cost_{-j}}\!\left[
        \frac{1+\theta_j}{\theta_j}p_j^s(\val,c_j,\cost_{-j})
        -
        c_jx_j(\val,c_j,\cost_{-j})
        \right]
        \ge
        \E_{\val,\cost_{-j}}\!\left[
        \frac{1+\theta_j}{\theta_j}p_j^s(\val,c'_j,\cost_{-j})
        -
        c_jx_j(\val,c'_j,\cost_{-j})
        \right];
    \]
    \item the truthful report satisfies the RoC constraint with complementary slackness:
    \[
        \E_{\val,\cost}\!\left[
        p_j^s(\val,\cost)-c_jx_j(\val,\cost)
        \right]\ge0
        \quad\perp\quad
        \theta_j\ge0.
    \]
\end{enumerate}
Then truthful reporting is optimal for seller \(j\) among all globally RoC-feasible deviations. Here
\[
    x_j(\val,c_j,\cost_{-j})=\sum_i x_{ij}(\val,c_j,\cost_{-j}).
\]
\end{lemma}

\begin{proof}
Let \(T\) denote truthful reporting, and let \(f\) be any globally RoC-feasible randomized deviation. Define
\[
    C_j^f
    :=
    \E_{\val,c_j,\cost_{-j}}\!\left[
    c_jx_j(\val,f(c_j),\cost_{-j})
    \right],
    \qquad
    P_j^f
    :=
    \E_{\val,c_j,\cost_{-j}}\!\left[
    p_j^s(\val,f(c_j),\cost_{-j})
    \right],
\]
and define \(C_j^T,P_j^T\) analogously for truthful reporting.

Taking expectation in condition (1) over \(c_j\), and extending to randomized deviations by linearity, gives
\[
    \frac{1+\theta_j}{\theta_j}P_j^T-C_j^T
    \ge
    \frac{1+\theta_j}{\theta_j}P_j^f-C_j^f.
\]
Since \(\theta_j>0\), complementary slackness in condition (2) implies \(P_j^T=C_j^T\). Since \(f\) is RoC-feasible, \(P_j^f\ge C_j^f\). Therefore,
\[
    \frac{1}{\theta_j}C_j^T
    \ge
    \frac{1+\theta_j}{\theta_j}P_j^f-C_j^f
    \ge
    \frac{1+\theta_j}{\theta_j}C_j^f-C_j^f
    =
    \frac{1}{\theta_j}C_j^f.
\]
Thus \(C_j^T\ge C_j^f\), so truthful reporting maximizes expected traded cost over all globally RoC-feasible deviations.
\end{proof}

\end{document}